\documentclass[a4paper]{scrartcl}

\usepackage[utf8]{inputenc}
\usepackage{amsthm}%
\usepackage{amssymb}%
\usepackage{enumerate}%
\usepackage{mathtools}%

\newcommand{\cupdot}{\mathbin{\mathaccent\cdot\cup}}

\usepackage{hyperref}
\usepackage{cleveref}

\usepackage{tikz}
\usetikzlibrary{hobby}%
\usetikzlibrary{calc}%
\usepackage{tikz-cd}%
\usepackage{bm}%
\usetikzlibrary{shapes}

\newcommand{\FPC}{\textsf{\upshape FPC}}
\newcommand{\CPT}{\textsf{\upshape CPT}}
\newcommand{\BGS}{\textsf{\upshape BGS}}
\newcommand{\WL}{\textsf{\upshape WL}}
\newcommand{\DWL}{\textsf{\upshape DeepWL}}

\newcommand{\FOH}{\textsf{\upshape FO+H}}

\newcommand{\IL}{\textsf{\upshape IL}}

\newcommand{\dom}{\operatorname{dom}}
\newcommand{\codom}{\operatorname{codom}}
\newcommand{\supp}{\operatorname{supp}}

\newcommand{\lef}{{\operatorname{left}}}
\newcommand{\rig}{{\operatorname{right}}}
\newcommand{\lt}{{\operatorname{le}}}
\newcommand{\rt}{{\operatorname{ri}}}
\newcommand{\pa}{{\operatorname{pa}}}
\newcommand{\norm}{{\operatorname{norm}}}

\newcommand{\init}{{\operatorname{init}}}
\newcommand{\step}{{\operatorname{step}}}
\newcommand{\halt}{{\operatorname{halt}}}
\newcommand{\out}{{\operatorname{out}}}

\newcommand{\scc}{{\operatorname{scc}}}
\newcommand{\sccs}{{\operatorname{SCC}}}
\newcommand{\pair}{{\operatorname{pair}}}
\newcommand{\pai}{{\texttt{pair}}}
\newcommand{\diag}{{\operatorname{diag}}}

\newcommand{\work}{{\operatorname{wk}}}
\newcommand{\query}{{\operatorname{ia}}}
\newcommand{\conf}{{\operatorname{cc}}}

\newcommand{\reg}{\textsf{\upshape C}}
\newcommand{\tap}{\textsf{\upshape T}}
\newcommand{\addcom}{\texttt{contract}}
\newcommand{\addunion}{\texttt{create}}
\newcommand{\addpair}{\texttt{addPair}}
\newcommand{\forget}{\texttt{forget}}
\newcommand{\enc}{\texttt{enc}}

\newcommand{\across}{{\operatorname{cross}}}
\newcommand{\plain}{{\operatorname{plain}}}

\newcommand{\image}{{\operatorname{image}}}

\let\downarrowOld\downarrow
\renewcommand{\downarrow}{\mathord{\downarrowOld}}
\let\uparrowOld\uparrow
\renewcommand{\uparrow}{\mathord{\uparrowOld}}

\renewcommand{\hat}{\widehat}
\renewcommand{\tilde}{\widetilde}
\renewcommand{\hat}{\widehat}

\theoremstyle{plain}
\newtheorem{theo}{Theorem}
\newtheorem*{theo*}{Theorem}
\crefname{theo}{Theorem}{Theorems}
\newtheorem{lem}[theo]{Lemma}
\crefname{lem}{Lemma}{Lemmata}
\newtheorem{cor}[theo]{Corollary}
\newtheorem*{cor*}{Corollary}
\crefname{cor}{Corollary}{Corollarys}

\theoremstyle{definition}

\newtheorem{exa}[theo]{Example}
\crefname{exa}{Example}{Examples}

\crefname{prob}{Problem}{Problems}

\crefname{open}{Open Problem}{Open Problems}

\theoremstyle{remark}
\newtheorem{rem}[theo]{Remark}

\crefname{section}{Section}{Sections}
\crefname{appendix}{Appendix}{Appendices}
\crefname{figure}{Figure}{Figures}

\definecolor{myBlue}{rgb}{0.5, 0.5, 1.0}
\definecolor{myRed}{rgb}{1.0, 0.5, 0.5}
\definecolor{myGreen}{rgb}{0.5, 1.0, 0.5}
\definecolor{myYellow}{rgb}{1.0, 1.0, 0.0}
\definecolor{myBlack}{rgb}{0.0, 0.0, 0.0}
\definecolor{myOrange}{rgb}{1.0, 0, 0}
\definecolor{myPurple}{rgb}{1.0, 0, 1.0}
\definecolor{myBaby}{rgb}{0, 1.0, 1.0}

\renewcommand{\phi}{\varphi}
\renewcommand{\epsilon}{\varepsilon}
\renewcommand{\theta}{\vartheta}

\newcommand{\NN}{{\mathbb N}}

\newcommand{\LC}{\textsf{\upshape C}}

\newcommand{\CA}{{\mathcal A}}

\newcommand{\CC}{{\mathcal C}}

\newcommand{\CE}{{\mathcal E}}

\newcommand{\CG}{{\mathcal G}}

\newcommand{\CI}{{\mathcal I}}

\newcommand{\CO}{{\mathcal O}}
\newcommand{\CP}{{\mathcal P}}

\newcommand{\CS}{{\mathcal S}}

\newcommand{\CV}{{\mathcal V}}

\newcommand{\CZ}{{\mathcal Z}}

\DeclareMathOperator{\Can}{Can}

\newcommand{\Pair}{\textsf{\upshape Pair}}
\newcommand{\Union}{\textsf{\upshape Union}}
\newcommand{\Card}{\textsf{\upshape Card}}
\newcommand{\Atoms}{\textsf{\upshape Atoms}}
\newcommand{\TheUnique}{\textsf{\upshape TheUnique}}
\newcommand{\Empty}{\textsf{\upshape Empty}}
\newcommand{\HF}{\operatorname{HF}}

\begin{document}

\title{Deep Weisfeiler Leman}        
\author{Martin Grohe\\
\normalsize RWTH Aachen University
\and
Pascal Schweitzer\\
\normalsize TU Kaiserslautern
\and
Daniel Wiebking\\
\normalsize RWTH Aachen University
}

\date{}
\maketitle
\begin{abstract}\sloppy
We introduce the framework of Deep Weisfeiler Leman algorithms (\DWL), which allows the design of purely combinatorial graph isomorphism tests that are more powerful than the well-known Weisfeiler-Leman algorithm.

We prove that, as an abstract computational model, polynomial time
\DWL-algorithms have exactly the same expressiveness as the logic Choiceless Polynomial Time (with counting) introduced by Blass, Gurevich, and Shelah (Ann. Pure Appl.  Logic., 1999)

It is a well-known open question whether the existence of a polynomial
time graph isomorphism test implies the existence of a polynomial time
canonisation algorithm. Our main technical result states that for each
class of graphs (satisfying some mild closure condition), if there is
a polynomial time \DWL\ isomorphism test then there is a polynomial
canonisation algorithm for this class. This implies that there is also
a logic capturing polynomial time on this class.
\end{abstract}

\theoremstyle{plain}
\newcommand{\thistheoremname}{}
\newtheorem*{namedthm}{\thistheoremname}
\newenvironment{thmblock}[1]
  {\renewcommand{\thistheoremname}{#1}%
   \begin{namedthm}}
   {\end{namedthm}}

\section{Introduction}

The research that lead to this paper grew out of the following
seemingly unrelated questions in the context of the graph isomorphism
problem.

\begin{thmblock}{Question~A}
Are there efficient combinatorial graph
  isomorphism algorithms more powerful than the standard Weisfeiler
  Leman algorithm?
\end{thmblock}

Here we are interested in general purpose isomorphism algorithms and
not specialised algorithms for specific graph classes.

\begin{thmblock}{Question~B}
Are there generic methods to construct graph canonisation algorithms from
isomorphism algorithms?
\end{thmblock}

This question is also related to an important open problem in
descriptive complexity theory, the question of whether there is a
logic \emph{capturing} polynomial time. Such a logic would express
exactly the properties of graphs that are polynomial-time
decidable. It is known that if there is a polynomial time canonisation
algorithm for a class of graphs then there is a logic that captures
polynomial time on that class. (The converse is
unknown.)

Initially, we studied Questions~A and B separately, but at some point,
we noted an interesting connection, which is based on the empirical
observation that typically combinatorial isomorphism algorithms can
easily be lifted to canonisation algorithms, whereas for group
theoretic algorithms this is not so easy. Before giving any details,
let us discuss the two questions individually.

\subsection*{From Weisfeiler Leman to DeepWL}
One of the oldest (and most often re-invented) graph isomorphism
algorithm is the \emph{colour refinement algorithm}, which is also
known as \emph{naive vertex classification} or \emph{1-dimensional
  Weisfeiler-Leman algorithm} (1-{\WL}). It iteratively colours the
vertices of a graph. Initially, all vertices get the same
colour. The initial colouring is repeatedly \emph{refined}, in the
sense that colour classes are split into several classes. In each
refinement round, two vertices that still have the same colour get different
colours in the refined colouring if they have a different number of
neighbours in some colour class of the current colouring. The refinement
process stops if no further refinement can be achieved; we call the
resulting colouring \emph{stable}. As such, 1-{\WL} just
computes a colouring of the vertices of a graph, but it can be used as
an isomorphism test by running it simultaneously on two graphs and
comparing the colour histograms. If there is some colour such that the
two graphs have a different number of vertices of this colour, we know
the graphs are non-isomorphic, and we say that 1-{\WL} 
\emph{distinguishes} the two graphs. 1-{\WL} is an
incomplete isomorphism test, that is, there are non-isomorphic
graphs not distinguished by the algorithm. The simplest example is a
cycle of length 6 versus two triangles. 

In order to design a more powerful isomorphism test, Weisfeiler and
Leman \cite{weilem68} proposed a similar iterative colouring procedure for pairs of
vertices; this led to what is now known as the \emph{classical} or
\emph{2-dimensional Weisfeiler-Leman algorithm} (2-{\WL}). In the initial
colouring, the colour of a pair $(u,v)$ indicates whether $u$ and $v$
are equal, adjacent, or distinct and non-adjacent. Then in each
refinement round, two pairs $(u,v)$ and $(u',v')$ that still have the
same colour $a$ get different colours if for some colours $b,c$ the
numbers of vertices $w$ and $w'$ in the configuration shown in
Figure~\ref{fig:2wl} are distinct.
Again, the refinement
process stops if no further refinement can be achieved. The algorithm
can easily be adapted to directed graphs,
possibly with loops and labelled edges. All we need to do is
modify the initial colouring. For example, if we have two edge labels
$R,S$, the initial colouring of $2$-{\WL} has twenty different colours
encoding the isomorphism types of pairs $(u,v)$, for example, ``$u=v$ and there is
an $R$-loop, but no $S$-loop on $u$'' or ``$u\neq v$, there is no edge
from $u$ to $v$, and there is both an $R$-edge and an $S$-edge from
$v$ to $u$''. Throughout this paper, it will be convenient for us to
work with edge-labelled directed graphs, that is, \emph{binary relational
  structures}. 

\definecolor{darkgreen}{rgb}{0.0, 0.42, 0.24}
\begin{figure}
  \centering

  \begin{tikzpicture}[
        vertex/.style = {draw,fill,circle,inner sep=0pt, minimum height
          = 2mm},
        ]
        \begin{scope}
          \node[vertex] (u) at (0,0) {}; 
          \node[vertex]  (v) at (1.5,0) {}; 
          \node[vertex]  (w1) at (0,1.5) {}; 
          \node[vertex]  (w2) at (1.5,1.5) {}; 
          \path (u) node[below] {$u$} (v) node[below] {$v$} (w1)
          node[above] {$w_1$} (w2) node[above] {$w_2$};

          \draw[very thick,blue,->] (u) edge node[below] {\color{blue}$a$} (v);
          \draw[very thick,red,->] (u) edge node[left,pos=0.3] {\color{red}$b$} (w1);
          \draw[very thick,darkgreen,->] (w1) edge node[right,pos=0.3]
          {\color{darkgreen}$c$} (v);
          \draw[very thick,red,->] (u) edge node[left,pos=0.3] {\color{red}$b$} (w2);
          \draw[very thick,darkgreen,->] (w2) edge node[right,pos=0.3] {\color{darkgreen}$c$} (v);
        \end{scope}
        \begin{scope}[xshift=4cm]
          \node[vertex] (u) at (0,0) {}; 
          \node[vertex]  (v) at (1.5,0) {}; 
          \node[vertex]  (w1) at (-0.25,1.5) {}; 
          \node[vertex]  (w2) at (0.75,1.5) {}; 
          \node[vertex]  (w3) at (1.75,1.5) {}; 
          \path (u) node[below] {$u'$} (v) node[below] {$v'$} (w1)
          node[above] {$w_1'$} (w2) node[above] {$w_2'$} (w3) node[above] {$w_3'$};

          \draw[very thick,blue,->] (u) edge node[below] {\color{blue}$a$} (v);
          \draw[very thick,red,->] (u) edge node[left,pos=0.3] {\color{red}$b$} (w1);
          \draw[very thick,darkgreen,->] (w1) edge node[right,pos=0.15]
          {\color{darkgreen}$c$} (v);
          \draw[very thick,red,->] (u) edge node[left,pos=0.4] {\color{red}$b$} (w2);
          \draw[very thick,darkgreen,->] (w2) edge node[right,pos=0.15] {\color{darkgreen}$c$} (v);
         \draw[very thick,red,->] (u) edge node[left,pos=0.35] {\color{red}$b$} (w3);
          \draw[very thick,darkgreen,->] (w3) edge node[right,pos=0.3] {\color{darkgreen}$c$} (v);
        \end{scope}
  \end{tikzpicture}

  \caption{2-{\WL} differentiates between $(u,v)$ and $(u',v')$ if for some colours~$b$ and~$c$ there
    are different numbers of vertices $w_i$ and~$w'_i$, respectively, such that $(u,w_i)$, $(u',w'_i)$ have
    colour $b$ and $(w_i,v)$, $(w'_i,v')$ have colour $c$.}
\label{fig:2wl}
\end{figure}
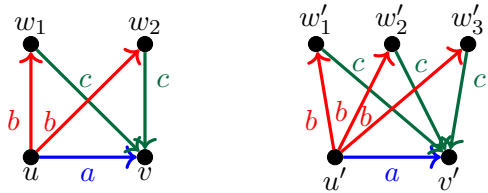

When it comes to distinguishing graphs, 2-{\WL} is significantly more
powerful than 1-{\WL}, but it is still fairly easy to find
non-isomorphic graphs not distinguished by the algorithm. In fact, any
two strongly regular graphs with the same parameters are
indistinguishable by $2$-{\WL}. To further strengthen the algorithm,
Babai proposed to colour $k$-tuples (for an arbitrary $k$)
instead of just pairs of vertices, introducing the
\emph{$k$-dimensional Weisfeiler-Leman algorithm} ($k$-{\WL}) (see \cite{caifurimm92}). 
For constant $k$, the algorithm runs in polynomial time, to be
precise the result is computable in time $O(n^{k+1}\log n)$.
This arguably still simple combinatorial algorithm
is quite powerful. It subsumes all natural combinatorial approaches to
graph isomorphism testing and, remarkably, also many algebraic and
mathematical optimisation approaches (e.g.~\cite{atsman13,atsoch18,bergro15,odowriwu+14}), with the important
exception of the group theoretic approaches introduced by Babai and
Luks \cite{bab79,babluk83,luk82} in the early 1980s.

It is quite difficult to find non-isomorphic graphs not
distinguishable by $k$-{\WL}, even for constant $k\ge 3$. In a seminal
paper, Cai, F\"urer and Immerman~\cite{caifurimm92} constructed, for every $k$, a pair
$G_k,H_k$ of non-isomorphic graphs of size $O(k)$ that are not
distinguished by $k$-{\WL}. These so-called \emph{CFI-graphs} encode the
solvability of a system of linear equations over a finite field, and
all known examples of non-isomorphic graph pairs not distinguished by
the Weisfeiler-Leman algorithms are based on variations of this
construction. Incidentally, the hardest known instances for practical
graph isomorphism tools are based on the same construction \cite{neuschwe17}. Let us
remark that the CFI graphs can easily be distinguished in polynomial
time by group theoretic techniques. Indeed, the graphs are 3-regular
and thus can be distinguished by Luks's \cite{luk82} polynomial time
isomorphism algorithm for graph classes of bounded degree. But the
group theoretic techniques are far more
complicated than the simple ``local constraint propagation''
underlying the Weisfeiler-Leman algorithm.

This brings us to Question~A. We start from a different perspective on
$k$-{\WL}: instead of colouring $k$-tuples, we can think of $k$-{\WL} as
adding all $k$-tuples of vertices as new elements to our input graph,
together with new binary relations encoding the relationship between
the tuples and vertices of the original graphs. Then on this extended
graph we run 1-{\WL} (or, depending on the details of the construction,
$2$-{\WL}), and the resulting colours of the $k$-tuples should
correspond to (or subsume) the colours $k$-{\WL} would assign to these
tuples. This correspondence between $k$-{\WL} on a graph and 1-{\WL} on an
extended structure consisting of $k$-tuples of vertices of the
original graph has been known for a while, it may go back to the work
of Otto~\cite{ott97}. \emph{Here is our new idea: }perhaps we do not need all
$k$-tuples of vertices to distinguish two graphs, but just a few of
them.  This could arise in a situation where we have two graphs $G,H$
and within them small subsets $S\subseteq V(G), T\subseteq V(H)$ such that the difference between
the graphs is confined to the induced subgraphs of these
subsets. Then to distinguish the graphs, it suffices to create tuples
of elements of these subsets.

\begin{exa}
  Let $G$ and $H$ be the graphs obtained
by padding the CFI-graphs $G_k,H_k$ with $2^{k\log k}$ isolated vertices,
and let
$S\subseteq V(G)$ and $T\subseteq V(H)$ be the vertex sets of
$G_k,H_k$ within $G,H$, respectively. 

Then $n=|G|=|H|=O(k)+2^{k\log k}$, and we need (at least) the $(k+1)$-{\WL} to
distinguish $G,H$, running in time $n^{\Omega(k)}=n^{\Omega(\frac{\log
  n}{\log\log n})}$. However, to distinguish the graphs, we only need to see all
$(k+1)$-tuples of vertices from the sets $S,T$, and the number of such
tuples is $k^{(k+1)}=O(n)$. Thus if we create only these $(k+1)$-tuples
and then use 1-{\WL} to distinguish the graphs extended by these
tuples, we have a polynomial time algorithm.
\end{exa}

The example nicely illustrates that it can be beneficial to confine
the use of a high-dimensional {\WL} algorithm to a small part of a
structure. It allows us to investigate this part to greater depth,
using $k$-{\WL} even for $k$ linear in the size of the relevant part
while maintaining an overall polynomial running time. The question is
how we find suitable sets $S$ and $T$ on which we focus. We can start
from the colour classes of 1-{\WL} on the current structure. Then we
can iterate the whole process, that is: we start by running 1-{\WL} on
the input graph(s), then choose one or several colours with few
elements, add $k$-tuples of elements of these colours, extend the
graph by these tuples and the associated relations, then run 1-{\WL}
again, choose new colour classes, add tuples, et cetera. We repeat
this procedure as long as our running time permits it.  This is the
idea of \emph{Deep Weisfeiler Leman ({\DWL})}, a class of
combinatorial algorithms that are based on the same simple
combinatorial ideas as Weisfeiler Leman, but turn out to be
significantly more powerful.

The formal realisation of this idea is subtle and requires some
care. Without going into too many details here (see
Section~\ref{sec:dwl}), let us highlight some of the main
points. First of all, since we can iterate the process of tuple
creation, it suffices to create pairs; $k$-tuples can be encoded as
nested pairs. Second, it turns out that working with
$2$-{\WL} instead of $1$-{\WL} leads to a much more robust class of
algorithms. One intuitive reason for this is that $2$-{\WL} (as opposed to $1$-{\WL})
allows us to trace connectivity and paths in a graph and thereby
allows us to detect if two deeply nested pairs share elements of the
input graph. On a technical level, 2-{\WL} allows us to use
the language and algebraic theory of coherent configurations \cite{DBLP:books/daglib/0037866},
which
are tightly linked to colourings computed by 2-{\WL}.
Moreover, the creation of pairs is particularly natural
in combination with $2$-{\WL}: we simply pick a colour class of the
current colouring (of pairs of elements of the current structure) and
then create a new element for each pair of that colour. 

A third aspect of the formalisation of {\DWL} is less intuitive, but
leads to an even more powerful class  of algorithms that is also more
robust (as \cref{theo:dwlcptpil} shows). Besides creating pairs of
elements, we introduce a second operation for \emph{contracting}
connected components of a colour class (or factoring a
structure). This allows us to discard irrelevant information and better control the size of the structure
we build. 

So what is a {\DWL} algorithm? Basically, it is a strategy for
adaptively choosing a sequence of operations (create elements
representing pairs, contract connected components) and the colour
classes to which these operations are applied. A run of such an
algorithm maintains a growing structure (the original structure
plus the newly created elements and relations), but the algorithm has
no direct access to the structure. It only gets the information of
which colours the 2-{\WL} colouring of the structure computes and how the
colour classes relate. This guarantees that a {\DWL} algorithm always
operates in an isomorphism invariant way: isomorphic input structures
lead to exactly the same runs.  We introduce {\DWL} as a general
framework for algorithms operating on graphs (and relational
structures), but we are mainly interested in graph isomorphism
algorithms that can be implemented in {\DWL}. It follows from our
results that {\DWL} can distinguish all CFI graphs in polynomial time
and thus is strictly more powerful than $k$-{\WL} for any $k$. (Note that $k$-{\WL} can
be seen as a specific {\DWL} algorithm where the strategy is to create
all $k$-tuples.)

\subsection*{Isomorphism Testing, Canonisation, and Descriptive
  Complexity}
The graph isomorphism problem can be seen as the algorithmic problem
of deciding whether two different representations of a graph, for
example, two different adjacency matrices, actually represent the same
graph. One way of solving this problem is to transform arbitrary
representations of a graph into a \emph{canonical
  representation}. A \emph{canonisation
  algorithm} does precisely this. Canonisation is an interesting
problem beyond isomorphism testing. For example, if we want to store
molecular graphs in a chemical information system, then it is
best to store a canonical representation of the molecules.

Formally, a \emph{canonical form} for a class $\CC$ of graphs (which
we assume to be closed under isomorphism) is a
mapping $\Can\colon\CC\to\CC$ such that for all $G\in\CC$, the graph $\Can(G)$ is
isomorphic to $G$, and for isomorphic $G,H\in\CC$ it holds that
$\Can(G)=\Can(H)$. A \emph{canonisation algorithm} for $\CC$ is an algorithm
computing a canonical form for $\CC$.\footnote{We view this as a
  promise problem, that is, it is irrelevant what the algorithm does
  on inputs $G\not\in\CC$.}

To the best of our knowledge, for all natural classes $\CC$ for which
a polynomial time isomorphism algorithm is known, a polynomial time
canonisation algorithm is also known. For some classes, for example
the class of planar graphs or classes of bounded tree width, it was
easy to generalise isomorphism testing to canonisation. For other
classes, for example all classes of bounded degree, this required
considerable additional effort \cite{babluk83}. Question~B simply asks
if there is a polynomial time reduction from canonisation to
isomorphism testing.  This is an old question (see, for example,
\cite{gur97}) that may eventually be resolved by a proof that there
exists a polynomial time canonisation algorithm for the class of all
graphs. But it is conceivable that this question can be resolved
without clarifying the complexity status of either isomorphism or
canonisation. In any case, it is consistent with current knowledge
that there is a polynomial time isomorphism algorithm, but no
polynomial time canonisation algorithm for the class of all graphs.

A pattern that emerged over the years is that it is usually easy to
obtain canonisation algorithms from combinatorial isomorphism
algorithms and much harder to obtain them from group theoretic
isomorphism algorithms. This intuition is supported by the following
(folklore) theorem: \textit{Let $\CC$ be a graph class such that
  $k$-{\WL} is a complete isomorphism test for $\CC$, that is, it distinguishes all non-isomorphic (vertex coloured) graphs in $\CC$. Then there is a
  polynomial time canonisation algorithm for $\CC$.} (For a proof, see
\cite{groneu19}.) Interestingly, for some graph classes for which
(group theoretic) polynomial time isomorphism tests were known, the
first polynomial time canonisation algorithms were obtained by proving
that $k$-{\WL} is a complete isomorphism test for these classes. Examples
are classes of bounded rank width \cite{groschwe15b,groneu19} and
graph classes with excluded minors \cite{pon88,gro17}.

As mentioned earlier, a polynomial time canonisation algorithm for a
class of graphs yields a logic that captures polynomial time on
that class. 
Arguably the most prominent
logic in this context is \emph{fixed-point logic with counting}
(\FPC)~\cite{imm87a,graott93}. \FPC\ captures polynomial
time on many natural graph classes, among them all classes with
excluded minors \cite{gro17}. There are deep connections between
$\FPC$ and the Weisfeiler-Leman algorithm. In particular, \emph{for
  every class $\CC$ of graphs, isomorphism for graphs from $\CC$ is
  expressible in \FPC\ if and only if there is a $k\ge 1$ such that
  $k$-{\WL} is a complete isomorphism test for $\CC$} \cite{ott97}.  A
consequence of this is that \FPC\ cannot express isomorphism of the
CFI graphs, which implies that the logic does not capture polynomial
time on the class of all graphs.

\emph{Choiceless polynomial time with counting} (\CPT) is a richer
logic that is strictly more expressive than \FPC, but still contained
in polynomial time (in the sense that all properties of graphs
expressible in \CPT\ are polynomial-time decidable). It was introduced
by Blass, Gurevich, and Shelah \cite{blagurshe99} as a formalisation of ``choiceless'',
that is, isomorphism invariant, polynomial time computations. Dawar,
Rossman, and Richerby \cite{dawricros08} proved that isomorphism of the
CFI graphs is expressible in \CPT. It is still an open question
if \CPT\ captures polynomial time.

\subsection*{Main Results}

Our first main result shows that polynomial time \DWL\ algorithms can
decide precisely the properties expressible in the logic \CPT.
Thus \DWL\ corresponds to \CPT\ in a similar way as the standard
WL-algorithm corresponds to the logic \FPC.

\begin{theo*}[\cref{theo:dwlcptpil}]
  A property of graphs is decidable by a polynomial time
  \DWL-algorithm if and only if it is expressible in \CPT.
\end{theo*}

\begin{cor*}
  There is a polynomial time {\DWL} algorithm that decides isomorphism
  of the CFI graphs.
\end{cor*}

A direct consequence of this result is that \DWL\ is strictly more powerful
than the standard WL-algorithm. Thus \DWL\ provides an answer to
Question~A: it gives us purely combinatorial isomorphisms tests
strictly more powerful than standard WL. Moreover, the logical
characterisation in terms of \CPT\ (\cref{theo:dwlcptpil}) shows that
the class of polynomial time \DWL-algorithms is robust and, arguably,
natural.

Our second main result addresses Question~B. While not fully resolving
it, it substantially extends the realm of isomorphism algorithms that
can automatically be transformed to canonisation algorithms. A
\emph{complete invariant} for a class $\CG$ of graphs is a mapping
$\CI\colon \CG\to\{0,1\}^*$ such that for all $G,H\in\CG$ we have $G\cong H$
if and only if $\CI(G)=\CI(H)$. 

\begin{theo*}[\cref{theo:inv}]
  Let $\CG$ be a class of graphs such that there is a polynomial-time
  {\DWL}-algorithm deciding isomorphism on $\CG$. Then there is a
  polynomial-time \DWL-algorithm that computes a complete invariant for $\CG$.
\end{theo*}

We say that a class $\CG$ of (vertex) coloured graphs is \emph{closed under
colouring} if all graphs obtained from a graph in $\CG$ by changing
the colouring also belong to $\CG$.
Isomorphisms between coloured graphs are defined in the usual way such that the colour
of each vertex has to be preserved. By a result due to
Gurevich~\cite{gur97} relating complete invariants to canonisation, we
obtain the following corollary.

\begin{cor*}[\cref{cor:canon}]
  Let $\CG$ be a class of coloured graphs closed under colouring such that
  there is a polynomial time {\DWL} algorithm deciding isomorphism on
  $\CG$. Then there is a polynomial time canonisation algorithm for
  $\CG$.
\end{cor*}

The rest of this paper is organised as follows. After giving the
necessary preliminaries in Section~\ref{sec:prel}, we formally
introduce \DWL\ in Section~\ref{sec:dwl}. In
  Section~\ref{sec:PureDWL}, we prove that \DWL\ is
  equivalent to a restricted form that we call pure
  \DWL. Section~\ref{sec:normDWL} is the technical core of the
  paper. We prove our main technical result about
  isomorphism testing in \DWL\ and canonisation. The difficult part of the
  proof is a normal form that we obtain for \DWL-algorithms deciding
  isomorphism.
  Finally, in Section~\ref{sec:cptdwl} we establish the equivalence
  between \DWL\ and \CPT. Due to space liminations, we have to defer
  many of the proofs to a technical appendix.

\section{Preliminaries}
\label{sec:prel}

\subsection*{Binary Relations and Structures}
Let $R$ be a binary relation.
The \emph{domain} of $R$ is defined as the set
$\dom(R):=\{u\mid\exists v\colon\;(u,v)\in R\}$, and
the \emph{codomain} of $R$ is
$\codom(R):=\{v\mid\exists u\colon\;(u,v)\in R\}$.
The \emph{support} of $R$
is $\supp(R):=\dom(R)\cup\codom(R)$.
The
\emph{converse} of $R$ is the relation $R^{-1}:=\{(v,u)\mid
(u,v)\in R\}$.
The
\emph{concatenation} of two binary relations $R_1,R_2$ is the relation
$R_1 \circ R_2:=R_1R_2:=\{(u,w)\mid\exists v\colon\;(u,v)\in R_1\text{ and }(v,w)\in
R_2\}$. Union, intersection and difference between relations are
defined in the usual set-theoretic sense. The \emph{strongly connected
  components} of a binary relation
$R$ are defined in the usual way as inclusionwise
maximal sets $S\subseteq \dom(R)\cap\codom(R)$ such that for all $u,v\in S$ there is an
$R$-path of length at least $1$ from $u$ to $v$.
(In particular a singleton
set $\{u\}$ can be a strongly connected component only if $(u,u)\in R$.)
We write $\sccs(R)$ to denote the set of strongly connected
components of $R$. Moreover, we let $R^\scc:=\bigcup_{S\in\sccs(R)}S^2$ be the relation describing whether two elements are in the same strongly connected component.
For a set $V$, the \emph{diagonal} of $V$ is the
relation $\diag(V):=\{(v,v)\mid v\in V\}$. For a relation $R$ we let
$R^\diag:=R\cap\diag(\dom(R))$ be the diagonal elements in~$R$. We call $R$ a 
\emph{diagonal relation} if $R=R^\diag=\diag(\dom(R))$.

A \emph{vocabulary} is a finite set $\tau$ of binary relation
symbols.  In some places, we need to specify how relation symbols
$R\in\tau$ are represented: we always assume that they are binary
strings $R\in\{0,1\}^*$. In particular, this allows us to order the
relation symbols lexicographically.  Note that the lexicographical
order on the relation symbols induces a linear order on each
vocabulary $\tau$. This will be important later, because it allows us
to represent vocabularies in a canonical way.  Let
$R_1,\ldots,R_t\in\{0,1\}^*$ be the sequence of all relation symbols
in $\tau$ according to the lexicographical order.  A
\emph{$\tau$-structure} $A$ is a tuple $(V(A),R_1(A),\ldots,R_t(A))$
consisting of a finite set $V(A)$, the \emph{vertex set} or
\emph{universe}, and a (possibly empty) relation
$R(A)\subseteq V(A)^2$ for each relation symbol $R\in\tau$.  We view
graphs as structures whose vocabulary consist of a single relation
symbol $E$. While we are mainly interested in graphs, to develop our
theory it will be necessary to consider general structures. We may see
structures as directed graphs with coloured edges; each binary
relation symbol corresponds to an edge colour and edges may have
multiple colours. Note that we can also simulate unary relations and
hence vertex colourings in binary structures $A$ by diagonal
relations.

Besides \emph{substructures} of a structure (obtained by
deleting vertices and edges) and \emph{restrictions} of a structure
(obtained by removing entire relations from the structure and the
vocabulary), we sometimes need to consider a combination of both. Let $A$ be a $\tau$-structure and let $\tilde\tau\subseteq\tau$ and $\tilde V\subseteq V(A)$.
The \emph{$\tilde\tau$-subrestriction} of $A$ on $\tilde V$ is the
$\tilde\tau$ structure $\tilde A:=A[\tilde\tau,\tilde V]$ with
universe $V(\tilde A)=\tilde V$ and $E(\tilde A)=E(A)\cap \tilde V^2$ for all
$E\in\tilde \tau$. 
We write $A[\tilde V]$ to denote $A[\tau,\tilde
V]$. 

The \emph{Gaifman graph} of a $\tau$ structure $A$ is the
undirected graph with vertex set $V(A)$ in which two elements $v,w$
are adjacent if they are related by some relation of $A$, that is,
$(v,w)\in R(A)$ or $(w,v)\in R(A)$ for some $R\in\tau$. A structure
$A$ is \emph{connected} if its Gaifman graph is connected.

Isomorphisms between $\tau$-structures are defined as bijective
mappings between their universes that preserve all relations. We write
$A\cong A'$ to denote that $A$ and $A'$ are isomorphic. Structures of
distinct vocabularies are non-isomorphic by definition. A
\emph{property} $\CP$ of structures is an
isomorphism closed class of structures. If all structures in $\CP$
have the same vocabulary $\tau$, then $\CP$ is a property \emph{of
$\tau$-structures}. An
\emph{invariant} for a class $\CC$ of structures (that we usually
assume to be closed under isomorphism) is a
mapping $\CI$ with domain $\CC$ such that $A\cong A'\implies
\CI(A)=\CI(A')$. If the converse also holds, that is, $A\cong A'\iff
\CI(A)=\CI(A')$, then $\CI$ is a \emph{complete invariant} for $\CC$. A
\emph{canonical form} is a complete invariant $\Can$ whose range also
consists of structures from $\CC$ and that satisfies $A\cong \Can(A)$ for all
$A$.

When carrying out computations on structures, we need to fix an
encoding by binary strings. One way of doing this is to first specify
the vocabulary, as a list of binary strings representing the relation
symbols, then the universe, also as a list of binary strings
representing the elements, and then the actual relations as lists of
pairs of strings. The details of this encoding are not
important. However, it is important to note that this encoding is
\emph{not} canonical: isomorphic structures may end up with different
string encodings. Moreover, the encoding depends on how we represent
the elements of the universe by binary strings, implicitly fixing a
linear order on the universe. Obviously, the output of an algorithm computing a
property or invariant of abstract structures must not depend on this
choice.

\subsection*{Coherent Configurations and the Weisfeiler-Leman Algorithm}
Let $\sigma$ be a vocabulary. A \emph{coherent $\sigma$-configuration} $C$ is a $\sigma$-structure $C$
with the following properties.
\begin{itemize}
  \item $\{R(C)\mid R\in\sigma \}$ is a partition of $V(C)^2$. In
    particular, all relations $R(C)$ must be nonempty.
  \item For each $R\in\sigma$ the relation $R(C)$ is either a
  subset of or disjoint from the diagonal $\diag(V(C))$.
\item For each $R\in\sigma$ there is an
$R^{-1}\in\sigma$ such that $R^{-1}(C) = R(C)^{-1}$.
\item For all triples $R_1,R_2,R_3\in\sigma$ there is a number
$q = q(R_1, R_2 ,R_3)\in\NN$ such that for all $(u,v)\in R_1(C)$ there are
exactly $q$ elements $w \in V(C)$ such that $(u,w)\in R_2(C)$ and $(w,v)\in
R_3 (C)$.
\end{itemize}
The numbers $q(R_1,R_2,R_3)$ are called the \emph{intersection numbers} of
$C$ and the function $q\colon\sigma^3\to\NN$ is called the
\emph{intersection function}.

We say that a coherent $\sigma$-configuration~$C$ is at least as
\emph{fine} as, or \emph{refines}, a
$\tau$-structure $A$ (we write $C\sqsubseteq A$) if $V(C)=V(A)$ and for each $R\in\sigma$
and each $E\in\tau$
it holds that $R(C)\subseteq E(A)$ or
$R(C)\subseteq V(A)^2\setminus E(A)$.
Conversely, we say that $A$ is
at least as \emph{coarse} as, or \emph{coarsens}, $C$.
Two coherent configurations $C,C'$ are \emph{equally fine}, written $C\equiv C'$, if
$C\sqsubseteq C'$ and $C'\sqsubseteq C$. In this case, the coherent structures are equal up to a renaming of the vertices and
the relation symbols.
We say that a coherent configuration $C$ is a \emph{coarsest coherent configuration refining} a
structure $A$ if $C\sqsubseteq A$ and $C'\sqsubseteq C$ for every coherent
configuration $C'$ satisfying~$C'\sqsubseteq A$.
If both $C,C'$ are coarsest coherent configurations refining $A$, then $C\equiv C'$.

\begin{theo}[\cite{immlan90,weilem68}]
  For every binary structure $A$ there is a coarsest coherent
  configuration $C$ refining $A$, and given $A$ it can be computed in
  polynomial time (time $\CO(n^3\log n)$, to be precise).
\end{theo}

A \emph{coherently $\sigma$-coloured $\tau$-structure}
is a pair $(A,C)$ consisting of a $\tau$-structure $A$ and a coherent
$\sigma$-configuration $C$ refining $A$ (and thus it holds $V(A)=V(C)$).
Unless explicitly stated otherwise, we always
assume that the vocabulary of $A$ is $\tau$ and the vocabulary of $C$
is $\sigma$, and we say that the vocabulary of $(A,C)$ is $(\tau,\sigma)$. 
We call the relation symbols in $\sigma$ \emph{colours},
whereas we keep calling the symbols in $\tau$ relation symbols. We usually
denote colours (from $\sigma$) by $R$ and relation symbols (from
$\tau$) by $E$.

We define the \emph{symbolic subset relation} of a coherently coloured
structure $(A,C)$ to be the binary relation
$\subseteq_{\sigma,\tau}=\{(R,E)\in\sigma\times\tau\mid R(C)\subseteq E(A)\}\subseteq
\sigma\times\tau$. We often omit the subscripts and just write
$R\subseteq E$ instead of $R\subseteq_{\sigma,\tau} E$.
 The
\emph{algebraic sketch}
of a coherently coloured
structure $(A,C)$ is the tuple
\[
D(A,C)=(\tau,\sigma,\subseteq_{\sigma,\tau},q)
\]
consisting of the 
vocabularies $\tau$, $\sigma$, the symbolic subset relation
$\subseteq_{\sigma,\tau}$, and the intersection function
$q\colon\sigma^3\to\NN$ of $C$.

The next lemma says that for all coherently coloured structures
$(A,C)$, we can choose a canonical coarsest coherent configuration $C(A)$ in the
set $\{C'\mid C'\sqsubseteq A\}$. 

\begin{lem}\label{lem:canData}
There is a polynomial-time algorithm that, for a given algebraic sketch $D(A,C')$,
computes the algebraic sketch of $D(A,C)$ of a coherently
coloured structure $(A,C)$ where $C$ is a canonical coarsest coherent
configuration of $A$. 
\end{lem}

The assertion that $C$ is \emph{canonical} means that
$C=C(A)$ only depends on $A$, i.e., for
algebraic sketches $D(A,C'),D(A,C'')$ the algorithm has the same
output.
We also write $D(A)$ to denote the algebraic sketch $D(A,C(A))$.

In fact, the previous lemma implies that we can choose
a string encoding for $D(A)=D(A,C(A))$ canonically. Formally, this means that we have a function $\enc$ mapping each structure $A$ to a binary string $\enc(A)$ representing $D(A)$ 
such that for isomorphic structures $A,A'$ we have
$\enc(D(A))=\enc(D(A'))$.
To obtain a canonical string encoding, we have to explain how to encode algebraic sketches.
Algebraic sketches are tuples consisting of sets and relations on
binary strings and natural numbers and as such can be encoded by
binary strings.
We encode the natural numbers using the unary representation.
With a unary representation the encoding size of the sketch of a coherently coloured structure $(A,C(A))$
and the encoding size $n:=|D(A)|$ are polynomially bounded in each other.
This will be useful later.

\section{Deep Weisfeiler Leman}
\label{sec:dwl}

A \DWL-algorithm is a 2-tape Turing machine $M$ with an additional
storage device $\reg_\conf$, called \emph{cloud},
 that maintains a coherently
coloured structure $(A,C(A))$. The machine has a work tape $\tap_\work$
and an \emph{interaction tape} $\tap_\query$
that allows a limited
form of interaction with the coherently coloured structure in the
cloud $\reg_\conf$.

The input of a \DWL-algorithm $M$ is a structure $A$
(the vocabulary $\tau$ of $A$ does not need to be fixed and can vary across the inputs).
For the starting configuration of
$M$ on input $A$, the machine is initialised with
the coherently coloured structure $(A,C(A))$ in the cloud
and with the algebraic sketch $D(A)=D(A,C(A))$ (canonically encoded as a string) on the interaction tape. The work tape is initially empty.
The Turing machine never has direct access to the structures in its
cloud, but it operates on relation symbols and vocabularies. (Recall
our assumption that relation symbols are binary strings.)

The Turing machine works as a standard 2-tape Turing machine.
Additionally, there are three particular transitions that can modify
the coherently coloured structure in the cloud. 
   For such transitions, the Turing machine writes a
relation symbol $X\in\tau\cup\sigma$ or a set of colours
$\pi\subseteq\sigma$ on the interaction tape and enters one of the
four states $q_\addpair$, $q_\addcom$, $q_\addunion$ and $q_\forget$.
We say that the Turing machine \emph{executes} $\addpair(X)$,
$\addcom(X)$, $\addunion(\pi)$, $\forget(X)$, respectively.  These
transitions modify the structure $A$ that is stored in the cloud.  In
particular, they can create new relations and possibly new elements
that are added to the structure.

\paragraph*{$\addpair(X)$.}
The state $q_\addpair$ can be entered while $X\in\tau\cup\sigma$
is written on the interaction tape. 
If $X=E\in\tau$ is a relation symbol, let $P:=E(A)$, otherwise if
$X=R\in\sigma$ is a colour, let $P:=R(C(A))$.  In this case, the
machine will add a fresh vertex to the universe for each of the pairs
contained in $P$.  Formally, we update $V(A) \leftarrow V(A)\cupdot P$
(where $\cupdot$ denotes the disjoint union operator which we assume
to be defined in some formally correct way, but we never worry about
the identity (or name) of the elements in the disjoint union).  Next,
we will create relations that describe how the fresh vertices relate
to the old universe.  We update
$\tau \leftarrow \tau\cup\{E_\lef,E_\rig\}$ and define $D_X$ to be the
lexicographically first binary string that is not already contained in
$\tau$ and then we update $\tau\leftarrow \tau\cupdot \{D_X\}$ again.
The relation $D_X$ describes the fresh vertices: $D_X(A):=\diag(P)$.
The relations $E_\lef(A),E_\rig(A)$ describe how the fresh vertices
relate to the old universe:
$E_\lef(A) \leftarrow E_\lef(A)\cup\{(u,(u,v))\in V(A)^2\mid (u,v)\in
P\}$ and
$E_\rig(A) \leftarrow E_\rig(A)\cup\{(v,(u,v))\in V(A)^2\mid (u,v)\in
P\}$ (in case that $E_\lef,E_\rig$ were not already defined, we
initialise $E_\lef(A),E_\rig(A)$ with the empty set before we take the
union).

\paragraph*{$\addcom(X)$.}
The state $q_\addcom$ can be entered while $X\in\tau\cup\sigma$
is written on the interaction tape. 
We will define a set $\CS:=\sccs(U)$ consisting of strongly connected
components.  If $X=E\in\tau$ is a relation symbol, let
$\CS:=\sccs(E(A))$, otherwise if $X=R\in\sigma$ is a colour, let
$\CS:=\sccs(R(C(A)))$.  Let $U:=V(A)\setminus\bigcup \CS$.  Next, we
will contract these components: we update
$V(A)\leftarrow U\cupdot \CS$.  Let $D_X$ be the lexicographically
first binary string that is not already contained in $\tau$ and update
$\tau\leftarrow \tau\cupdot\{D_X\}$.  The relation $D_X$ describes the
fresh vertices: $D_X(A):=\diag(\CS)$.  We update the relations for
each $E\in\tau$ and set
$E(A)\leftarrow (E(A)\cap U^2)\cup\{(u,S)\mid \exists v\in
S\in\CS\colon\; (u,v)\in E(A)\}\cup \{(S,v)\mid \exists u\in
S\in\CS\colon\; (u,v)\in E(A)\}\cup\{(S_1,S_2)\mid\exists u\in S_1\in
\CS\exists v\in S_2\in \CS\colon\; (u,v)\in E(A)\}$.
  
\paragraph*{$\addunion(\pi)$.}
The state $q_\addunion$ can be entered while $\pi\subseteq\sigma$ is
written on the
interaction tape.
Let $E_\pi$ be the lexicographically first binary string that is not
already contained in $\tau$ and then update
$\tau:=\tau\cupdot\{E_\pi\}$ where
$E_\pi(A):=\bigcup_{R\in\pi}R(C(A))$.
  
\paragraph*{$\forget(X)$.}
The state $q_\forget$ is entered while $X=E\in\tau$ is written on the interaction tape.
   We update $\tau\leftarrow \tau\setminus\{E\}$.
   
Each of these four transitions therefore modify the structure $A$ in the cloud.
After such a transition, the machine recomputes the coarsest coherent configuration $C(A)$ refining $A$.
The coherently coloured structure $(A,C(A))$ is stored in the cloud and
the algebraic sketch $D(A)$ (canonically encoded as a
string) is written on the interaction tape.

\medskip

Let us define the running time of \DWL-algorithms.
Recall that the input of the underlying Turing machine
is the algebraic sketch $D(A)$.
For the running time we take the following costs into account.
Each transition taken by the Turing machine counts as one time step.
For an input structure $A$, we take $n$-many steps into account to write down the initial algebraic sketch $D(A)$ to the tape (where $n=|D(A)|$ is the encoding length of $D(A)$).
Recall, that the intersection numbers are encoded using unary representation and therefore
each \DWL-algorithm needs at least linear time (in $|V(A)|$).
Similar, we also take $n'$-many steps into to write down the updated algebraic sketch $D(A')$ to the interaction tape
(where $n'=|D(A')|$ is the encoding length of $D(A')$).
We say that a \DWL-algorithm $M$ runs in polynomial time if there is a polynomial $p$ in $n=|D(A)|$ that bounds the running time of $M$.
The definition of polynomial time remains unchanged if we take polynomial costs into account for maintaining the cloud (such as
the running time of the Weisfeiler-Leman algorithm).

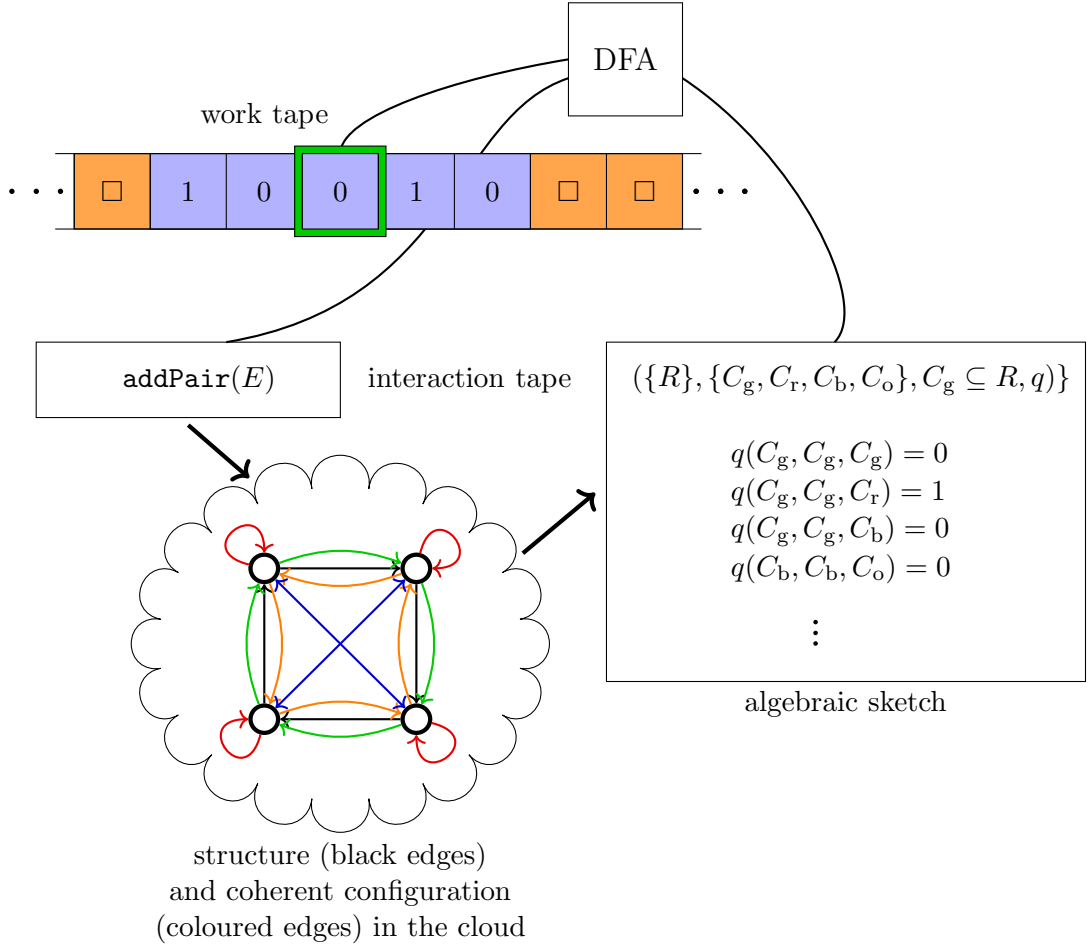
\begin{figure}
\centering
\begin{tikzpicture}[scale=1]

\draw  (2.5,7.5) rectangle (4,6);
\node at (3.25,6.75) {\large DFA};
\draw[thick] (2.5,6.75) .. controls (1,6.5) and (-0.5,6) .. (-0.5,5.5);
\draw[thick] (2.5,6.5) .. controls (1,6) and (1,3.5) .. (-2,3);
\draw[thick] (4,6.5) .. controls (5.5,5.5) and (6.5,3.5) .. (6,3);

\draw[fill=orange!70]  (-4,5.5) rectangle (4,4.5);
\draw[fill=blue!30]  (-3,5.5) rectangle (2,4.5);
\draw (-4.25,5.5) -- (4.25,5.5);
\draw (-4.25,4.5) -- (4.25,4.5);

\draw[fill=green!80!black]  (-1.1,5.6) rectangle (0.1,4.4);
\draw[fill=blue!30]  (-1,5.5) rectangle (0,4.5);

\node at (-1.5,6) {work tape};

\foreach \x in {-4,-3,-2,-1,0,1,2,3,4}
{
\draw (\x,5.5) -- (\x,4.5);
}

\foreach \x/\text in {-3.5/\square,-2.5/1,-1.5/0,-0.5/0,0.5/1,,1.5/0,2.5/\square,3.5/\square}
{
\node at (\x,5) {$\text$};
}
\node at (4.5,5) {\huge$\ldots$};
\node at (-4.5,5) {\huge$\ldots$};

\node[cloud, cloud puffs=20,cloud puff arc=200, aspect=2, draw, minimum width=5.5cm, minimum height = 5cm
    ]  at (-0.5,-1) {};

\node[text width=5.5cm,text centered] at (-0.5,-4.3) {structure (black edges) and coherent configuration
  (coloured edges) in the cloud};
\draw  (3,3) rectangle (9.3,-1.5);
\node at (6.15,-1.8) {algebraic sketch};
\draw  (-4.5,3) rectangle (-0.5,2);
\node at (1.2,2.5) {interaction tape};
\node[anchor=west] at (3.2,2.5) {$(\{R\},\{C_{\text{g}}, C_{\text{r}}, C_{\text{b}}, C_{\text{o}}\}, C_{\text{g}}\subseteq R,q)\}$};

\node[draw, circle, fill=white,ultra thick] (v2) at (0.5,0) {};
\node[draw, circle, fill=white,ultra thick] (v3) at (0.5,-2) {};
\node[draw, circle, fill=white,ultra thick] (v4) at (-1.5,-2) {};
\node[draw, circle, fill=white,ultra thick] (v1) at (-1.5,0) {};

\foreach \x/\y in {1/2,2/3,3/4,4/1}
{
\draw[thick, ->] (v\x) edge (v\y);
\draw[thick, ->, green!80!black] (v\x) edge[bend left=20] (v\y);
\draw[thick, <-, orange] (v\x) edge[bend left=-20] (v\y);

}
\draw[thick, <->, blue!80!black] (v4) -- (v2);
\draw[thick, <->, blue!80!black] (v1) -- (v3);

\draw[thick, ->, red!90!black] (v1) .. controls (-2.5,0.25) and (-1.5,1) .. (v1);
\draw[thick, ->, red!90!black] (v2) .. controls (0.75,1) and (1.5,0) .. (v2);

\draw[thick, ->, red!90!black] (v3) .. controls (1.5,-2.25) and (0.5,-3) .. (v3);
\draw[thick, ->, red!90!black] (v4) .. controls (-1.75,-3) and (-2.5,-2) .. (v4);
\node[anchor=west] at (4.5,1.5) {$q(C_\text{g},C_\text{g},C_\text{g})=0$};
\node[anchor=west] at (4.5,1) {$q(C_\text{g},C_\text{g},C_\text{r})=1$};
\node[anchor=west] at (4.5,0.5) {$q(C_\text{g},C_\text{g},C_\text{b})=0$};
\node[anchor=west] at (4.5,0) {$q(C_\text{b},C_\text{b},C_\text{o})=0$};

\node at (4.5,5) {\huge$\ldots$};
\node at (5.75,-0.75) {\huge$\vdots$};

\node[anchor =west] at (-3.5,2.5) {$\addpair(E)$};

\draw[ultra thick,->] (-2.5,1.9) -- (-1.7,1.2);

\draw[ultra thick,->] (1.9,0.2) -- (2.85,1);

\end{tikzpicture}%
\caption{A \DWL\ Turing machine}
\end{figure}

\subsection*{Isomorphism Invariance}
Whenever a \DWL-computa\-tion creates new elements, there is an issue
with the identity (or name) of these new elements. This problem is implicit in
our usage of the disjoint union operator in the definitions. We do not
have to worry about this, because \DWL-algorithms never have direct
access to the structure in the cloud and only depend on its
isomorphism types. All we need to make sure is that in every step the newly created
elements are distinct from the existing ones. 

Since it will be very important throughout the paper, let us state the
invariance condition more formally. We define the \emph{internal run} of a
\DWL-algorithm $M$ on input $A$ to be the sequence of configurations of the
underlying Turing machine without the content of the cloud.
Then the internal run of a \DWL-algorithm is 
isomorphism invariant: if $A$ and $A'$ are isomorphic $\tau$-structures, then $D(A)=D(A')$ and
therefore the internal run of a \DWL-algorithm $M$ on input $A$ is identical with
the internal run of $M$ on input $A'$. The contents of the clouds in
the computations of $M$ on inputs $A$ and~$A'$ may differ, but they are isomorphic
in the corresponding steps of the computations.

\subsection*{\DWL-Decidability and Computability}
A~\DWL-algorithm \emph{accepts} an input if the algorithm halts with 1 written under the head on the work tape.
A \DWL-algorithm \emph{decides} a property $\CP$ of structures
if it accepts an input structure $A$ whenever $A\in\CP$ and rejects $A$
otherwise.
(We do not require that all structures in $\CP$ have the same vocabulary.)

A \DWL-algorithm \emph{computes} a function $\CI\colon\CA\to\{0,1\}^*$ for a
class $\CA$ of structures if on input $A\in\CA$ it
stops with $\CI(A)\in\{0,1\}^*$ on the work tape. Observe that if
$\CI\colon\CA\to\{0,1\}^*$ is
\DWL-computable then it is an invariant, that is,  if $A,A'\in\CA$ are
isomorphic then $\CI(A)=\CI(A')$.

Let $\CE$ be a function that maps each structure $A$ to a binary
relation $\CE(A)$ over $V(A)$. A \DWL-algorithm
\emph{computes} $\CE$ if on input $A\in\CA$ it stops with a coherently
$\sigma'$-coloured $\tau'$-structure $(A',C(A'))$ and
the encoding of a relation symbol $E'\in\tau'$ on the work tape
such that $E'(A')=\CE(A)$.

Let $\CS$ be a function that maps each structure $A$ to  a subset $\CS(A)\subseteq V(A)$.
A \DWL-algorithm \emph{computes} $\CS$ if it computes $\CE(A):=\diag(\CS(A))$.

\subsection*{Basic \DWL{}-computable functions}
\label{sec:basic}
In the following three lemmas, we collect a few basic properties and
functions that are \DWL-computable.

\begin{lem}\label{lem:op1}
Let $E_1,E_2\in\{0,1\}^*$ be two relation symbols. Then the following
functions on the class of all $\tau$-structures with $E_1,E_2\in\tau$ are
\DWL-computable in polynomial time.
\begin{enumerate}
  \item\label{e:cup} $\CE_\cup(A):=E_1(A)\cup
  E_2(A)$.
  \item\label{e:cap} $\CE_\cap(A):=E_1(A)\cap
  E_2(A)$.
  \item\label{e:setminus} $\CE_\setminus(A):=E_1(A)\setminus E_2(A)$.
  \item\label{e:diag} $\CE_\diag(A):=\diag(V(A))$.%
  \item\label{e:conv} $\CE_{-1}(A):=E_1(A)^{-1}$.%
  \item\label{e:circ} $\CE_\circ(A):=E_1(A)\circ E_2(A)$.
  \item\label{e:scc} $\CE_\scc(A):=E_1(A)^\scc$.
\end{enumerate}
Moreover, all the above functions are computable by polynomial-time
\DWL-algorithms that do not add any new vertices to the structure.
\end{lem}

\begin{lem}\label{lem:op2}
Let $E_1,E_2\in\{0,1\}^*$ be two relation symbols. Then the following
properties of $\tau$-structures with $E_1,E_2\in\tau$ are \DWL-decidable in polynomial time.
\begin{enumerate}
  \item\label{e:subseteq} $\CP_{\subseteq}:=\{A\mid E_1(A)\subseteq E_2(A)\}$
    and $\CP_=:=\{A\mid E_1(A)= E_2(A)\}$.
  \item\label{e:=} $\CP_{\|\leq\|}:=\{A\mid |E_1(A)|\leq |E_2(A)|\}$ and $\CP_{\|=\|}:=\{A\mid |E_1(A)|= |E_2(A)|\}$.
\end{enumerate}
Moreover, all the above properties are decidable by polynomial-time
\DWL-algorithms that do not add any new vertices to the structure.
\end{lem}

\begin{lem}\label{lem:op3}
Let $E\in\{0,1\}^*$ be a relation symbol. Then the following
functions $\CS$ on the class of all $\tau$-structures with $E\in\tau$ are
\DWL-computable in polynomial time.
\begin{enumerate}
  \item\label{e:dom} $\CS(A):=\dom(E(A))$.
  \item\label{e:codom} $\CS(A):=\codom(E(A))$.
  \item\label{e:supp} $\CS(A):=\supp(E(A))$.
\end{enumerate}
Moreover, all the above functions are computable by polynomial-time
\DWL-algorithms that do not add any new vertices to the structure.
\end{lem}

\section{Pure Deep Weisfeiler Leman}
\label{sec:PureDWL}

A \DWL-algorithm is called \emph{pure} if the algorithm executes
$\addpair(R)$ and $\addcom(R)$ only for colours $R\in\sigma$ and not
for relation symbols $E\in\tau$. It is sometimes convenient to only
consider pure \DWL-algorithms, and it is an indication of the
robustness of the definition that every \DWL-algorithm is equivalent
to a pure one.

\begin{theo}\label{theo:pure}
Let $\CE$ be a function that assigns each structure $A$ a relation $\CE(A)$
and that is \DWL-computable in polynomial time.
Then there is a pure \DWL-algorithm that computes $\CE$ in polynomial-time.
(The same holds for functions $\CI:\CA\to\{0,1\}^*$ and properties $\CP:\CA\to\{0,1\}$.)
\end{theo}

\section{Normalised Deep Weisfeiler Leman}
\label{sec:normDWL}

This section focuses on \DWL-algorithms $M$ that decide isomorphism on
a class $\CC$.  There are actually two natural ways of using \DWL{} to
decide isomorphism.

The first is the obvious one if we take \DWL{} as a computation model:
an isomorphism test implemented in \DWL{} takes the disjoint union of
two structures as its input and then decides if they are
isomorphic. For simplicity, we assume that the two structures are
connected; we can always achieve this by adding a new relation that
relates all elements of the respective structures. The \emph{disjoint
  union} of two $\tau$-structures $A_1$ and $A_2$ is the
$\tau$-structure $A:=A_1\uplus A_2$ with universe
$V(A)=V(A_1)\cupdot V(A_2)$ and relations $E(A)=E(A_1)\cupdot E(A_2)$
for $E\in\tau$.  A \DWL-isomorphism test takes the disjoint union
$A_1\uplus A_2$ of two connected $\tau$-structures as its input and
decides the property of the two components of its input structure
being isomorphic. We say that a \DWL-algorithm \emph{decides
  isomorphism on a class $\CC$} of structures if it correctly decides
isomorphism for structures $A_1,A_2\in\CC$.

The second way of using \DWL{} to decide isomorphism is inspired by
the way the classical Weisfeiler-Leman algorithm is used to decide
isomorphism:
the WL algorithm is said to distinguish two
structures $A_1$ and $A_2$ if the algebraic sketches $D(A_1)$ and $D(A_2)$
are distinct. Note that this is an incomplete isomorphism test, because two
structures may have the same algebraic sketch even though they are non-isomorphic.
Generalising this notion, we say that a \DWL-algorithm \emph{distinguishes} two
$\tau$-structures $A_1$ and $A_2$ if on input $A_i$
the algorithm halts with coherently coloured structure $(A_i',C(A_i'))$ in the cloud such that $D(A_1')\neq D(A_2')$. 
(Note that here the algorithm only takes a single structure as
its input.) If $A_1$ and $A_2$ are isomorphic, it can never
happen that a \DWL-algorithm distinguishes them, because \DWL-computations are isomorphism invariant. However,
as for the classical Weisfeiler-Leman algorithm, there may be
non-isomorphic structures not distinguished by the algorithm.
We say
that a \DWL{}-algorithm is a \emph{distinguisher for a class $\CC$} of structures if
it distinguishes all non-isomorphic structures $A_1,A_2\in\CC$.

In this section, we shall prove that for each class $\CC$ of
structures, there is a polynomial-time \DWL-algorithm deciding
isomorphism on $\CC$ if and only if there is a polynomial-time
\DWL-algorithm that is a distinguisher for $\CC$
(Theorem~\ref{theo:inv}). Thus the two concepts we just defined agree.
This is a nontrivial theorem that is an important step towards a
central goal of the paper: to turn isomorphism testing
algorithms into canonisation algorithms (Question~B).
The reason why
the equivalence between isomorphism test and distinguisher is
important is that a canonisation algorithm, just like a distinguisher,
only works on one input structure, whereas an isomorphism test works
on two structures simultaneously.  In fact, this issue lies at the
heart of the isomorphism versus canonisation problem. So, in order to
modify a \DWL-isomorphism test we first have to decouple the
computations on the two structures from each other, that is, transform
an isomorphism test into a distinguisher.

The main
technical complication in doing so is that a \DWL-algorithm working on
the disjoint union of two graphs can create pairs with endpoints in
each of the structures and then create pairs of such pairs and so on,
which makes an association of the created objects with one of the
structures impossible. We will call algorithms that avoid such
constructions normalised. The largest part of this section is devoted
to a proof that it suffices to consider normalised algorithms.

We say that $v\in V(A_1\uplus A_2)$ \emph{belongs} to $A_i$ if $v\in V(A_i)$ for $i\in \{1,2\}$.
We consider structures~$A$ that are iteratively obtained from $A_1\uplus A_2$ by applying the \DWL-algorithm $M$.
Inductively, if $M$ adds a pair $(u,v)$ to the universe where both $u,v$ belong the same structure $A_i$, then we say
that the fresh vertex $(u,v)$ \emph{belongs} to $A_i$ for $i\in \{1,2\}$.
Analogously, if $M$ contracts a component $S\in\sccs(R(C(A)))$ where all vertices in $S$ belong to the same structure $A_i$,
then we say that the fresh vertex $S$ \emph{belongs} to $A_i$ for $i\in \{1,2\}$.
We define $\CV_i(A):=\{v\in V(A)\mid v\text{ belongs to }A_i\}$ for both $i\in \{1,2\}$.
The vertices $\CV_\plain(A):=\CV_1(A)\cup\CV_2(A)$ are called \emph{plain}.
The edges $\CE_\plain(A):=\CV_1(A)^2\cup\CV_2(A)^2$ are called \emph{plain}.
The edges $\CE_\across(A):=\{(v_i,v_j)\mid v_i\in \CV_i(A),v_j\in\CV_j(A),i\neq j\in\{1,2\}\}$ are called \emph{crossing}.
We write $\sigma_\across$, to denote the set of colours $R\in\sigma$ such that $R(C(A))\subseteq\CE_\across(A)$.
Analogously, we define $\sigma_\plain$.
Note that it might be the case, that a relation
is neither a subset of $\CE_\plain(A)$ nor $\CE_\across(A)$.
In such a case, the relation is neither plain nor crossing. The goal is to avoid the construction of such relations.

A \DWL{}-algorithm $M$ that decides isomorphism on a class $\CC$
is called \emph{normalised} if
at any point in time
$V(A)=\CV_\plain(A)$ and
$A$ consists of exactly two connected components ($\CV_1(A)$ and
$\CV_2(A)$).
We also say that a structure~$A$ is \emph{normalised} if it is 
obtained from a normalised \DWL-algorithm.

Since a normalised structure $A$ consists of exactly two connected
components during the entire run, its holds that
$\sigma=\sigma_\plain\cupdot\sigma_\across$ for vocabulary of the coherent configuration $C(A)$.
Moreover, $\CE_\plain$ and~$\CE_\across$ are \DWL-computable functions
since a \DWL-algorithm can detect whether a colour
belongs to edges between the same and between distinct connected components.
Since all relation symbols are plain, we can see $A$ as a disjoint union of
its subrestrictions $A[\CV_1(A)]$ and $A[\CV_2(A)]$. The following
lemma gives a connection between
the coherent configuration $C(A)$ and its subrestrictions. 

\begin{lem}\label{lem:ccNorm}
Let $A$ be a $\tau$-structure that is normalised.
Let $\sigma_i\subseteq\sigma$ be the set of colours
$R$ such that $R(C(A))\cap \CV_i(A)^2\neq \emptyset$.
Let $A_i:=A[\CV_i(A)]$ be the
$\tau$-subrestriction, and let $C_i:=C(A)[\sigma_i,\CV_i(A)]$ be the
$\sigma_i$-subrestriction, for $i \in \{1,2\}$.
\begin{enumerate}
  \item\label{e:product} It holds $\sigma=(\sigma_1\cup\sigma_2)\cupdot\sigma_\across$ and for each $R\in\sigma_\across$,
there are diagonal colours $R_1\in\sigma_1,R_2\in\sigma_2$
such that $R(C(A))=\{(v_1,v_2)\in\CE_\across(A)\mid v_i\in\supp(R_i(C))\}$,
  \item\label{e:ind} $(A_i,C_i)$ is again a coherently coloured structure.

Moreover, given the set $\{D(A_1,C_1),D(A_2,C_2)\}$, the algebraic sketch $D(A)$
can be computed in polynomial time,
\item\label{e:eqfine} The coherent configurations $C_i$ and~$C(A_i)$  are equally fine (i.e., $C_i\equiv C(A_i)$).

Moreover, when given the set $\{D(A_1),D(A_2)\}$, the set $\{D(A_1,C_1),D(A_2,C_2)\}$
can be computed in polynomial time.
\end{enumerate}
\end{lem}

By $A_R$, we denote the structure that is obtained from $A$ by executing $\addcom(R)$ for some colour $R\in\sigma$.
The next lemma tells us that, given $D(A)$, the algebraic sketch $D(A_R)$
can be computed (by a Turing machine).
In this sense, the contraction of strongly connected components does not lead to new information.
However, it can still be useful, since it shrinks the size of the universe which might help
to ensure a polynomial bound on the universe size.

\begin{lem}\label{lem:comData}
Let $A$ be a $\tau$-structure.
\begin{enumerate}
  \item\label{com:alg} There is a
polynomial-time algorithm that for a given algebraic sketch $D(A)$ and a given
colour $R\in\sigma$, computes the algebraic sketch
$D(A_R)$.
\item\label{com:fine} The coherent configuration $C(A)[V(A)\setminus\bigcup\CS]$ refines $C(A_R)[V(A_R)\setminus\CS]$
where $\CS:=\sccs(R(A))$.
\end{enumerate}
\end{lem}

The next lemma is of a similar flavour.
It states that we do not get any information if we apply $\addpair(R)$
to a coherently coloured structure that is normalised where
$R\in\sigma_\across$.
By $A^R$, we denote the structure that is obtained from $A$ by executing $\addpair(R)$ for a colour $R\in\sigma$ with $R\subseteq\CE_\across(A)$.
More generally, we write $A^\omega$ to denote the structure that is obtained from $A$ by executing $\addpair(R)$ for all colours $R\in\omega\subseteq\sigma$.

\begin{lem}\label{lem:pairData}
Let $A$ be a $\tau$-structure that is normalised and let $\omega\subseteq\sigma_\across$ be a set of crossing colours.
\begin{enumerate}
  \item\label{pair:alg} There is a
polynomial-time algorithm that for a given algebraic sketch
$D(A)$ and a given $\omega$, computes the algebraic
sketch $D(A^\omega)$.
\item\label{pair:eqfine} The coherent configurations $C(A)$ and~$C(A^\omega)[V(A)]$ are equally fine (i.e., $C(A)\equiv C(A^\omega)[V(A)]$).
\end{enumerate}
\end{lem}

Towards showing that all \DWL-algorithm can be assumed to be normalised we take an intermediate step of algorithms that are almost normalised.

\subsection*{Almost Normalised \DWL{}}
In the following, let $M$ be a \DWL-algorithm that decides isomorphism on $\CC$.
If $M$ adds a pair $v=(v_1,v_2)$ to the universe where
$(v_1,v_2)\in\CE_\across(A)$, then that fresh vertex $v$ is called a \emph{crossing} pair.
The set of crossing vertices is denoted by $\CV_\across(A)$.
Let $\CV_\pair(A)$ denote the set of all pair-vertices (created by adding
pairs).

A \DWL{}-algorithm $M$ that decides isomorphism on a class $\CC$
is called \emph{almost normalised} if
each created pair-vertex is either plain or crossing,
i.e., at every point in time
$\CV_\pair(A)\subseteq\CV_\plain(A)\cup\CV_\across(A)$.
Furthermore, we require for
$\forget(E)$-executions
that $E\notin\{E_\lef,E_\rig\}\cup\tau$.
This additional requirement ensures that
$\CE_\plain,\CE_\across,\CV_\plain,\CV_\across$ are \DWL-computable functions.

We say that a structure $A$ is \emph{almost normalised} if it is
computed by an almost normalised \DWL-algorithm.

\begin{lem}\label{lem:almostNorm}
Assume that isomorphism on some class $\CC$ is \DWL-decidable in polynomial time.
Then, isomorphism on $\CC$ is decidable by an almost normalised \DWL{}-algorithm in polynomial time.
\end{lem}

The proof of the lemma shows how to simulate a \DWL-algorithm~$\hat M$ with an almost normalised algorithm~$M$. While technically quite involved, the idea behind it is simple. To simulate~$\hat M$, we need to guarantee that all vertices that are created are actually crossing or plain, that is, they are formed from pairs of plain vertices. For this, we continuously guarantee that for each vertex~$v$ created so far, there are two unique plain vertices~$p_\lt(v),p_\rt(v)$ that are only associated with~$v$.
Instead of creating a new pair~$(v,v')$ for two arbitrary vertices we will then create the pair~$(p_\lt(v),p_\rt(v'))$ instead. The crucial point is that this way the new vertex is still plain or crossing because $p_\lt(v)$ and~$p_\rt(v)$ are plain. 

There is a second less severe issue that we need to take care of, namely vertices that are crossing could lose this property if the two plain end-points of which they are the crossing pair are contracted away. For this reason we copy vertices and use the originals when applying $\addpair$-operations. We use copies to create pairs but only ever contract originals, which resolves the second issue.
For details we refer to Appendix~\ref{app:normDWL}.

\begin{lem}\label{lem:norm}
Assume that isomorphism on some class $\CC$ is \DWL-decidable in polynomial time.
Then, isomorphism on $\CC$ is decidable by a normalised \DWL{}-algorithm in polynomial time.
\end{lem}

Similarly to the previous lemma, the proof of this lemma shows how to simulate a \DWL-algorithm that is almost normalised with a normalised \DWL-algorithm. The crucial insight here is Lemma~\ref{lem:pairData} which essentially says that we do not need to execute an $\addpair(R)$-operation when the relation~$R$ is crossing, since we can compute the algebraic sketch of the result without modifying the contents of the cloud.
One might be worried that we cannot continue to simulate the algorithm because the structure in the cloud was not modified, but 
\cref{lem:pairData} Part~\ref{pair:eqfine} shows that the coherent configuration does not become finer by any of the $\addpair$-operations we only simulate.  A similar statement holds for contractions by \cref{lem:comData} Part~\ref{com:fine}. Again we refer to Appendix~\ref{app:normDWL} for details.

\begin{theo}\label{theo:inv}
Let $\CC$ be a class of $\tau$-structures. 
The following statements are equivalent.
\begin{enumerate}
  \item\label{e:iso} Isomorphism on $\CC$ is \DWL-decidable in polynomial time.
  \item\label{e:inv} Some complete invariant on $\CC$ is \DWL-computable in polynomial time.
  \item\label{e:dis} There is a \DWL-distinguisher for $\CC$ that runs in polynomial time.
\end{enumerate}
\end{theo}

\begin{proof}
To see that \ref{e:dis} implies \ref{e:iso}, let
$M$ be an distinguisher for $\CC$.
Given $A=A_1\uplus A_2$, we simulate the distinguisher on $A_1,A_2$ in parallel.
If the algebraic sketches $D(A_1)$ and $D(A_2)$ are distinct at some point in time, we
stop the simulation and reject isomorphism.
Otherwise, if the algebraic sketches coincide during the entire run, we accept.
Since $M$ is a distinguisher, the structures $A_1$ and $A_2$ are
isomorphic.

To see that \ref{e:inv} implies  \ref{e:dis}, let $M$ be a
\DWL-algorithm that computes a complete invariant on $\CC$. Let
$A_1,A_2$ be non-isomorphic. Then the runs $\rho_1,\rho_2$ of $M$ on inputs $A_1,A_2$
are distinct. The only way this can happen is that at some point
during the computation the algebraic sketches are distinct. We need to
make sure that they remain distinct until the end of the
computation. We can achieve this by never executing a
$\forget$-operation and by making a copy of the elements that are
deleted during a $\addcom$-operation before carrying out the
contraction. Copying elements can easily be done by applying
$\addpair$ to the corresponding diagonal colours. Since we never
discard elements or relations, once distinct, the algebraic sketches can never
become equal again. Of course now we always
have a richer structure with a potentially finer coherent
configuration in the cloud, but we can still simulate the original
computation of $M$.

It remains to show that \ref{e:iso} implies \ref{e:inv}.
Let $M$ be a \DWL-algorithm that decides isomorphism on $\CC$ in polynomial-time, i.e.,
it accepts $A=A_1\uplus A_2$ if and only if $A_1,A_2$ are isomorphic.
By \cref{lem:norm}, there is a normalised \DWL-algorithm $M_\norm$
that decides isomorphism on $\CC$.
We let $\rho_\norm(A)\in\{0,1\}^*$ denote the internal run of the \DWL-algorithm $M_\norm$ on input $A$.
We claim that the function $\CI(A_1):=\rho_\norm(A_1\uplus A_1)$ defines
a complete invariant that is \DWL-computable in polynomial time.
To compute $\CI$ with a \DWL-algorithm, we can easily simulate
$M_\norm$ on input $A_1\uplus A_1$ and track its internal run.
We show that $\CI$ defines a complete invariant.
Clearly, if $A_1$ and $A_2$ are isomorphic, then $\CI(A_1)=\CI(A_2)$.
On the other hand,
assume that $\CI(A_1)=\CI(A_2)$.
By definition of $\CI$, we have $I:=\rho_\norm(A_1\uplus A_1)=\rho_\norm(A_2\uplus A_2)$.
By \cref{lem:ccNorm}, we have that $I=\rho_\norm(A_1\uplus A_2)$
since $D(A)=D(A')$ if and only if $\{D(A[\CV_i(A)])\mid i\in \{1,2\}\}=\{D(A'[\CV_i(A')])\mid i\in \{1,2\}\}$
holds during the entire run.
Since $I$ is an accepting run and since $M_\norm$ decides isomorphism, it follows that $A_1,A_2$
are isomorphic.
\end{proof}

We say that a class $\CG$ of vertex-coloured graphs is \emph{closed under
colouring} if all structures obtained from a graph in $\CG$ by changing
the colouring also belong to $\CG$.
Formally, a vertex-coloured graph is a structure $A$ with vocabulary $\tau=\{E,C_1,\ldots,C_t\}$
where $E(A)$ is a binary relation and $C_i(A)$ are diagonal relations (which encode the $i$-th vertex colour).
A class $\CG$ is closed under vertex-colouring if for each $\{E,C_1,\ldots,C_t\}$-structure $A\in\CG$
also the $\{E,C_1',\ldots,C_s'\}$-structure $A'$ belongs to $\CG$ where $E(A)=E(A')$ and $C_i'(A')$ are arbitrary diagonal relations.

\begin{theo}[\cite{gur97}]
Let $\CG$ be a class of coloured graphs that is closed under colouring.
The following statements are equivalent.
\begin{enumerate}
  \item Some complete invariant on $\CC$ is computable in polynomial time.
  \item Some canonical form on $\CC$ is computable in polynomial time.
\end{enumerate}
\end{theo}

\noindent
This gives us the following corollary.

\begin{cor}\label{cor:canon}
Let $\CG$ be a class of coloured graphs closed under colouring such that
there is a polynomial time {\DWL} algorithm deciding isomorphism on
$\CG$. Then there is a polynomial time canonisation algorithm for
$\CG$.
\end{cor}

\begin{rem}
  It can be shown that \cref{cor:canon} also holds for arbitrary
  time-constructible time bounds of the form $T^{\CO(1)}$ where $T$ is
  at least linear.  This can be shown using padding arguments.
\end{rem}

\section{Equivalence between \CPT{} and \DWL{}}\label{sec:cptdwl}

In this section, we prove that \DWL\ has the same expressiveness as
the logic \CPT, choiceless polynomial time with counting.  As opposed
to the previous sections, in this section we use fairly standard
techniques from finite model theory. For this reason, and to avoid
lengthy definitions, we only describe the
high level arguments and omit the technical details. %

Let us start with a very brief introduction to \CPT, which follows the
presentation of \cite{gragro15}. We start by introducing a language
that we call $\BGS$ (for Blass, Gurevich, Shelah, the authors of the
original paper introducing choiceless polynomial time). The syntactical
objects of \BGS\ are \emph{formulas} and \emph{terms}; the latter play
the role of algorithms of \DWL. Like \DWL-algorithms, \BGS-terms operate on
a $\tau$-structure $A$, the \emph{input structure}. The terms are
interpreted over the hereditarily finite sets over $V(A)$, that is,
elements of $V(A)$, sets of elements of $V(A)$, sets of sets of
elements, et cetera.
We denote the set of all these hereditarily
finite sets by $\HF(A)$. \emph{\BGS-formulas} are Boolean combinations
of formulas $t_1=t_2$, $t_1\in t_2$, and $R(t_1,t_2)$, where
$t_1,t_2$ are terms and $R$ is in the vocabulary of the input
structure. The formulas $t_1=t_2$ and $R(t_1,t_2)$ have the obvious
meaning, $t_1\in t_2$ means that $a_1\in\HF(A)$
interpreting $t_1$ is an element of $a_2\in\HF(A)$
interpreting $t_2$.  \emph{Ordinary \BGS-terms} are formed using the
binary function symbols $\Pair$, $\Union$, the unary function symbols
$\TheUnique$, $\Card$, and the constant symbols $\Empty$ and
$\Atoms$. The constants are interpreted by the empty set and the set
$V(A)$, respectively. (The elements of $V(A)$ are the ``atoms'' or
``urelements'' in $\HF(A)$.) To
define the meaning of the function symbols, let $t_1$, $t_2$ be terms
that are interpreted by $a_1,a_2\in\HF(A)$. Then $\Pair(t_1,t_2)$ is
interpreted by $\{a_1,a_2\}$ and $\Union(t_1,t_2)$ is interpreted by
$a_1\cup a_2$, where we treat the union with an atom like the union
with the empty set. $\TheUnique(t_1)$ is interpreted by $b$ if
$a_1=\{b\}$ is a singleton set and by $\emptyset$
otherwise. $\Card(t_1)$ is interpreted by the cardinality of $a_1$
viewed as a von-Neumann ordinal.  Besides the ordinary terms,
BGS-logic has \emph{comprehension terms} of the form
\begin{equation}\label{eq:comp}
  \tag{$\star$}
\{ t\colon x\in u\colon \phi\}.
\end{equation}
Here $t,u$ are terms, $x$ is a variable that is not free in $u$, and
$\phi$ is a formula. To emphasise the role of the variable $x$, we
also write
$\{ t(x)\colon x\in u\colon \phi(x)\}$. This term
is interpreted by  the set of all values $t(a)$, where $a\in \HF(A)$ is an element of the
set defined by the term $u$, and $a$ satisfies the formula $\phi(x)$. Note
that $t,u$, and $\phi$ may have other free variables besides
$x$. The third type of \BGS-terms are \emph{iteration terms}. For
every term $t$ that has exactly one free variable $x$ we form a new
term $t^*$. We also write $t(x)$ and $t(x)^*$ to emphasise the role of
the free variable $x$. The term $t^*$ is
interpreted by the fixed point of the sequence
$t(\emptyset),t(t(\emptyset),t(t(t(\emptyset))),\ldots$, if such a
fixed point exists, or by $\emptyset$ if no fixed point exists.
The \emph{free variables} of terms
and formulas are defined in the natural way, stipulating that the variable $x$ in
a comprehension term of the form \eqref{eq:comp} and in an iteration term
$t(x)^*$ be bound. As usual, a \emph{sentence} is a formula without
free variables.

\CPT\ is the ``polynomial-time fragment'' of \BGS. To define \CPT\
formally, we restrict the length of iterations to be polynomial in the
size of the input structure, and we restrict the number of elements of
all sets that appear during the ``execution'' of a term to be
polynomial.

This completes our short description of the syntax and semantics of
$\BGS$ and $\CPT$. For details, we refer the reader to \cite{gragro15} (or
\cite{blagurshe99,dawricros08} for other, equivalent, presentations of
the language).

We say that a property $\CP$ of $\tau$-structures is \emph{\CPT-definable}
if there is a \CPT-sentence $\phi$ satisfied precisely by the
$\tau$-structures in $\CP$.

\begin{lem}\label{lem:dwlTOcpt}
  Let $\CP$ be a property of $\tau$-structures that is decidable by a
  polynomial time \DWL-algorithm. Then $\CP$ is $\CPT$-definable.
\end{lem}

\begin{proof}
  We show that we can simulate a polynomial time $\DWL$-algorithm in
  $\CPT$. Formally, for this we have to define a \CPT-term that
  simulates a single computation step and then iterate this term, but
  we never go to this syntactic level. We only describe the
  high-level structure of the simulation by explaining how we can
  simulate the different steps carried out during a \DWL-computation. 

  A basic fact that we use in this proof is that $\CPT$ is at least as
  expressive as fixed-point logic with counting
  \FPC~\cite{blagurshe99}. Moreover, we use the Immerman-Vardi
  Theorem~\cite{imm82,var82} stating that on ordered structures, \FPC\ (even
  fixed-point logic without counting) can simulate all polynomial time computations. 

  In our \CPT-simulation, we first generate the set $N$ of the first $n$
  von-Neumann ordinals, where $n$ is the size of the input
  structure, together with the natural linear order $\le$ of these ordinals.
  We use the ordered domain $(N,\le)$ to simulate the computation carried out by our \DWL-algorithm. Since we are considering a
  polynomial time computation, we can do this in \FPC\ and hence in
  \CPT.

  Moreover, during the computation we always maintain a copy of the
  structure $A$ in the cloud. The elements as well as the relations of
  $A$ are represented by hereditarily finite sets. To represent pairs
  of elements, we use the usual von-Neumann encoding of ordered pairs.
  We can represent the relations of $A$ as sets of pairs, which
  are again hereditarily finite sets. To maintain the name of a
  relation, for every relation we create a pair $(x,r)$ consisting of
  the relation, represented by a set $x$, and its name $r$, which is a
  von-Neumann ordinal whose binary representation is the name of the
  relation.

  Now an $\addpair$ step can easily be simulated. For the $\addcom$-operation, we
  note that strongly connected components are definable in fixed-point
  logic. For $\addunion$, we use the $\Union$-operation, and $\forget$
  is trivial.

  It remains to explain how we deal with the coherent configurations
  that the \DWL-algorithm maintains during its computation. For this,
  we use the fact that the coarsest coherent configuration of a
  structure, that is, the ordered partition of the pairs of elements
  into the colour classes computed by the 2-dimensional Weisfeiler
  Leman algorithm, is definable in the logic
  \FPC~\cite{graott93,ott97}. Thus we can compute the coherent
  configuration during our simulation, and using the $\Card$-operation
  we can extract the algebraic sketch.  This enables us to simulate
  all parts of a \DWL-computation in \CPT.
\end{proof}

To prove the converse direction that $\DWL$ can decide all
\CPT-definable properties, rather than
working with \CPT\ directly, we will use the equivalent
\emph{interpretation logic} introduced by Grädel, Pakusa, Schalthöfer
and Kaiser~\cite{grapakschalkai15}.

Interpretation logic is based on the logic $\FOH{}$, first-order logic with a \emph{Härtig} (or
\emph{equal cardinality}) quantifier: the meaning of the
formula $\textsf{H}x.\phi_1(x).\phi_2(x)$ is that the number of elements satisfying $\phi_1$ equals the number of elements
satisfying $\phi_2$. Next, we need to define a specific form of syntactical
interpretation. 
A (2-dimensional)
\emph{$[\tau,\tau']$-interpretation} (over \FOH{}) is a tuple
\[
  \CI\!\!=\!\!\big(\phi_V(x_1,x_2),\!\phi_=((x_1,x_2),\!(x_1',x_2')),\!
  \phi_E((x_1,x_2),\!(x_1',x_2'))_{E\in\tau'}\big)
\]
of $\FOH{}$-formulas.
Such an interpretation maps each $\tau$-struc\-ture $A$ to a $\tau'$-structure
$\CI(A)$ defined as follows. Let $V:=\{(v_1,v_2)\in V(A)^2\mid
A\models \phi_V(v_1,v_2)\}$, and let 
$\sim$ be the finest equivalence
relation on $V$ such that
  $(v_1,v_2)\sim(v_1',v_2')$ for all $(v_1,v_2),(v_1',v_2')\in V$
  with $A\models\phi_=((v_1,v_2),(v_1',v_2'))$. Then the universe of $
\CI(A)$ is $V/_\sim$, and for every relation symbol $E\in\tau'$, the
relation $E(\CI(A))$ contains all pairs $(a,a')\in (V/_\sim)^2$ such
that there are $(v_1,v_2)\in a$, $(v_1',v_2')\in a'$ with
$A\models\phi_E(v_1,v_2,v_1',v_2')$. 

Now we are ready to define interpretation logic \IL.
An $\IL{}[\tau]$-program is a tuple
$\Pi=(\CI_\init,\CI_\step,\phi_\halt,\phi_\out)$ where for some
vocabulary $\tau_\work$, 
$\CI_\init$ is a $[\tau,\tau_\work]$-interpretation, $\CI_\step$ is a $[\tau_\work,\tau_\work]$-interpretation,
and $\phi_\halt$ and $\phi_\out$ are $\tau_\work$-sentences.
Let $A$ be a $\tau$-structure.
A \emph{run} of $\Pi$ on $A$ is a sequence of $\tau_\work$-structures $A_1,\ldots,A_m$ where $A_1:=\CI_\init(A)$
and $A_{i+1}:=\CI_\step(A_i)$ and $A_i\nvDash\phi_\halt$ for all $1\leq i< m$
and $A_m\vDash\phi_\halt$.
The program \emph{accepts} $A$ if $A_m\vDash\phi_\out$, otherwise it \emph{rejects}.
We say that a program $\Pi$ \emph{decides} a property $\CP$ of $\tau$-structures
if it accepts an input structure $A$ if $A\in\CP$ and rejects $A$
otherwise.

We say that a program $\Pi$ runs in polynomial time
if $m+\sum_{1\leq i\leq m}|V(A_i)|$ is bounded by a polynomial.

\begin{theo}[\cite{grapakschalkai15}]\label{theo:pilTOcpt}
  For every property $\CP$ of $\tau$-structures,
  the following statements are equivalent.
\begin{enumerate}
  \item $\CP$ is decidable by a polynomial-time \IL{}-program.
  \item $\CP$ is \CPT-definable.
\end{enumerate}
\end{theo}

\noindent
The next lemma states that we can simulate $\FOH$-formulas in \DWL.

\begin{lem}\label{lem:foh2dwl}
  Let $\phi(x_1,x_2)$ be an $\FOH[\tau]$-sentence. Then the function
  mapping each $\tau$-structure $A$ to the binary relation
  $\phi(A):=\{(v_1,v_2)\in V(A)^2\mid A\models \phi(x_1,x_2)\}$ is
  computable by a polynomial time \DWL-algorithm.
\end{lem}

\begin{lem}\label{lem:pilTOdwl}
Let $\CP$ be a $\tau$-property of structures that 
is decidable by a polynomial-time $\IL[\tau]$-program.  Then $\CP$ is
decidable by a polynomial time \DWL-algorithm.
\end{lem}

\begin{proof}
Let $\Pi=(\CI_\init,\CI_\step,\phi_\halt,\phi_\out)$ be an \IL{}-program that decides $\CP$
in polynomial time.
We have to simulate the run $A_1,\ldots,A_m$ with a \DWL-algorithm.

To simulate a single application of $\CI_\init$ or $\CI_\step$ applied
to a structure $A$, we proceed as follows: using an
$\addpair$-operation, we first compute the extension $A'$ of $A$ by
all pairs of elements. Now we can translate the formulas of the
interpretation, which have either $2$ or $4$ free variables, to
formulas with $2$ free variables ranging over $A'$. Then we can apply
\cref{lem:foh2dwl} to compute the image of the interpretation. Note
that to factorise the structure by the equivalence relation $\sim$, we
need to apply a $\addcom$-operation.

By repeatedly applying this construction, we can simulate the whole run
of $\Pi$. To check the halting condition and generate the output, we
note that we can formally regard the \FOH-sentences $\phi_\halt$ and
$\phi_\out$ as formulas with 2 free dummy variables (these formulas
are either satisfied by all pairs of elements or by none). Then we
can apply \cref{lem:foh2dwl} again.
\end{proof}

\begin{theo}\label{theo:dwlcptpil}
  For every property $\CP$ of $\tau$-structures,
  the following statements are equivalent.
  \begin{enumerate}
  \item\label{e:dwl} $\CP$ is decidable by a polynomial-time
    \DWL-algorithm.
  \item\label{e:cpt} $\CP$ is $\CPT$-definable.
\end{enumerate}
\end{theo}

\begin{proof}
We combine \cref{lem:dwlTOcpt}, \cref{theo:pilTOcpt}, and \cref{lem:pilTOdwl}.
\end{proof}

\section{Concluding Remarks}
\DWL\ is a framework for combinatorial graph isomorphism tests
generalising the Weisfeiler Leman algorithm of any dimension. \DWL\ is strictly more
powerful than Weisfeiler Leman; in particular, it can distinguish the
so-called CFI-graphs that are hard examples for Weisfeiler Leman. We
show that \DWL\ has the same expressivity as the logic \CPT\
(choiceless polynomial time). In some sense, \DWL\ can be seen as an
isomorphism (or equivalence) test corresponding to \CPT\ in a similar
way that Weisfeiler Leman is an isomorphism test corresponding to the
logic \FPC\ (fixed-point logic with counting).

The definition of \DWL\ is fairly
robust, as our result that all \DWL\ algorithms are equivalent to more
restricted ``pure''
\DWL-algorithms shows.

As our main technical result, we prove that if \DWL\ decides
isomorphism on a class $\CC$ of graphs in polynomial time, then $\CC$
admits a polynomial time canonisation algorithm, and hence there is a logic
capturing polynomial time on $\CC$.

\newpage\appendix 

\noindent
{\LARGE\bfseries Appendix}

\section{Proofs and Details Omitted from Section~\ref{sec:prel}}

\begin{proof}[\textbf{Proof of Lemma~\ref{lem:canData}}]
We define a series~$\preceq_0,\preceq_1,\ldots$ of finer and finer total quasiorders of $\sigma$.

For each $R\in\sigma$, define
$\pi_R:=\{E\in\tau\mid R\subseteq_{\sigma,\tau} E\}$. Recall that we
assume the relation symbols to be binary strings. Thus the sets $\pi_R$ are
finite sets of binary strings and as such can be ordered
lexicographically. We let $R\preceq_0 R'$ if $\pi_R$ is
lexicographically smaller or equal to $\pi_{R'}$.
This gives us a total quasiorder.
We also write $R\equiv_0 R'$ if $R\preceq_0 R'$ and $R'\preceq_0 R$.
The equivalence class of $R\in\sigma$ is denoted by
$[R]_0:=\{R'\in\sigma\mid R'\equiv_0 R\}$.
Note that the set of equivalence classes $\sigma_0=\{[R]_0\mid R\in\sigma\}$ is linearly ordered.

Inductively, for each $R\in\sigma$, we define
\[\pi_{R,i+1}:=\{([R_2]_i,[R_3]_i,\sum_{\mathclap{R_2'\in[R_2]_i,R_3'\in[R_3]_i}}q(R,R_2',R_3'))\mid
R_2,R_3\in\sigma\}.\]
Since the set of equivalence classes and the set of natural numbers are both linearly ordered sets,
one can define a linear order also on the triples $([R_2]_i,[R_3]_i,q)$ contained in the sets $\pi_{R,i+1}$.
We let $R\preceq_{i+1} R'$ if $\pi_{R,i+1}$ is smaller or equal to $\pi_{R',i+1}$, i.e.,
it holds $|\pi_{R,i+1}|<|\pi_{R',i+1}|$ or $\pi_{R,i+1}=\pi_{R',i+1}$, or it holds
$\pi_{R,i+1}\neq\pi_{R',i+1}$ and $|\pi_{R,i+1}|=|\pi_{R',i+1}|$ and the smallest triple in
$\pi_{R,i+1}\setminus\pi_{R',i+1}$ is smaller than the smallest triple in
$\pi_{R',i+1}\setminus\pi_{R,i+1}$.

Let $m\in\NN$ be the smallest number for which $\preceq_m=\preceq_{m+1}$.
Now, ``$\preceq_m$'' defines a total linear order on the equivalence classes $[R]_m$.
We claim that $\sigma_m:=\{[R]_m\mid R\in\sigma\}$
with $[R]_m(C'):=\bigcup_{R'\in[R]_m}R'(C')$ defines a coherent configuration
refining $A$.
Setting 
\[q([R_1]_m,[R_2]_m,[R_3]_m)=\sum_{R_2'\in[R_2]_m,R_3'\in[R_3]_m}q(R_1,R_2',R_3')\]
is well defined (since $\preceq_m=\preceq_{m+1}$) and gives a coherent configuration.
It refines $A$ since $\sigma_0=\{[R]_0\mid R\in\sigma\}$ with $[R]_0(C'):=\bigcup_{R'\in[R]_0}R'(C')$
already refines $A$ by definition.

Next, we use this linear order to define a canonical set of relation symbols $\sigma_A$.
We define $\sigma_A$ as the set consisting of the first $|\sigma|$
binary strings that are not already contained in $\tau$ (w.r.t.~to the lexicographic order on strings).

We have to show that for algebraic sketches $D(A,C')$ and $D(A,C'')$, we obtain the same result.
Observe that the equivalence classes $[R]_i$ in each step might depend on the coherent configuration $C'$. However,
it can easily shown by induction that the set of relations $\{\bigcup_{R'\in[R]_i}R'(C')\mid R\in\sigma\}$
does not depend on the choice of the coherent $\sigma$-configuration $C'$.
\end{proof}

\section{Proofs and Details Omitted from Section~\ref{sec:dwl}}

\begin{proof}[\textbf{Proof of Lemma~\ref{lem:op1}}]
The \DWL-algorithm uses the command $\addunion(\pi)$ where $\pi$ depends
on the relation we want to compute. The resulting relation symbol $E_\pi$ is written to the
work tape. The set $\pi\subseteq\sigma$ is specified using the algebraic sketch as
follows.

\ref{e:cup}. For $\CE_\cup$, let $\pi:=\{R\in\sigma\mid R\subseteq
E_1\text{ or }R\subseteq E_2\}$.
  
\ref{e:cap}. For $\CE_\cap$, let $\pi:=\{R\in\sigma\mid R\subseteq
E_1$ and $R\subseteq E_2\}$.

\ref{e:setminus}. For $\CE_\setminus$, let $\pi:=\{R\in\sigma\mid R\subseteq
E_1$ and $R\not\subseteq E_2\}$.

\ref{e:diag}. For $\CE_\diag$, let $\pi:=\{R\in\sigma\mid\forall R',R''\in\sigma,R'\neq R'':\; q(R',R'',R)=0\}$. We show that $R\in\pi$, if and only if
$R$ is a diagonal colour. Suppose $R\in\pi$ and
$(u,v)\in R(C)$ and let $R'\in\sigma$ be a colour with $(u,u)\in
R'(C)$. By the choice of $R'$, it holds that
$q(R,R',R)\geq 1$. Since $R\in\pi$, it holds that
$R=R'$ is a diagonal colour.
On the other hand, assume that $R\in\sigma$ is a diagonal colour.
Let $R',R''\in\sigma$ be colours with $R\neq R'$. Since $R(C),R'(C)$ are disjoint, it
holds that $q(R',R'',R)=0$. Hence, $R\in\pi$.

\ref{e:conv}. For $\CE_{-1}$, let $\pi:=\{R\in\sigma\mid\exists
R_1\in\sigma,R_1\subseteq E_1\exists R_2\in\sigma,R_2\text{ diagonal}:\;q(R_2,R_1,R)\geq 1\}$
(a \DWL-algorithm can decide whether a colour $R\in\sigma$ is diagonal using \ref{e:diag}).
We show that $R\in\pi$, if
and only if $R(C)\subseteq E_1(A)^{-1}$. Assume that $R\in\pi$.
Equivalently, this means that there is a colour
$R_1\in\sigma$ with $R_1\subseteq E_1$ and a diagonal colour $R_2\in\sigma$
such that $q(R_2,R_1,R)\geq 1$. This means that for all $(u,u)\in
R_2$ there is a $w\in V(C)$ such that $(u,w)\in R_1(C)\subseteq E_1(A)$ and $(w,u)\in R(C)$. Therefore,
$R(C)\cap E_1(A)^{-1}\neq\emptyset$. Since $(A,C)$ is a coherently
coloured structure, this is equivalent to $R(C)\subseteq E_1(A)^{-1}$.
Conversely, if~$R(C)\subseteq E_1(A)^{-1}$ and~$(u,v)\in R(C)$ then choose~$R_2$ so that it contains~$(v,v)$ and~$R_1$ so that it contains~$(v,u)$, then~$R_1\subseteq E_1$
and $R_2$ is diagonal and~$q(R_2,R_1,R)\geq 1$ to see that~$R\in\pi$.

\ref{e:circ}. For $\CE_\circ$, let $\pi:=\{R\in\sigma\mid \exists
R_1,R_2\in\sigma,R_1\subseteq E_1,R_2\subseteq E_2:\;q(R,R_1,R_2)\geq 1\}$.

\ref{e:scc}. For $\CE_\scc$, observe that $E_1(A)^\scc=E_1(A)^m\cap
(E_1(A)^m)^{-1}$ where $E_1(A)^1:=E_1(A)$ and $E_1(A)^{i+1}:=E_1(A)^i\cup
E_1(A)E_1(A)^i$ and $m$ is the smallest number with $E(A)^{m+1}=E(A)^m$.
Using, \ref{e:cup},\ref{e:cap},\ref{e:conv} and \ref{e:circ}, the relation $E_1(A)^\scc$ can be
computed.
\end{proof}

\begin{proof}[\textbf{Proof of Lemma~\ref{lem:op2}}]
To evaluate the queries, we use the algebraic sketch as follows.

\ref{e:subseteq}. Observe that $E_1(A)\subseteq E_2(A)$, if and only if
for all $R\in\sigma$ it holds that $R\subseteq E_1$ implies $R\subseteq E_2$.
Moreover, $E_1(A)= E_2(A)$, if and only if $E_1(A)\subseteq E_2(A)$ and
$E_2(A)\subseteq E_1(A)$.

\ref{e:=}. We show that the machine $M$ can actually
compute the cardinality by using the following formula.
\begin{align*}
|E_1(A)|&=\sum\limits_{\mathclap{\substack{R_1\in\sigma\\
R_1\subseteq E_1}}}|R_1(C)|\\
\text{where }|R_1(C)|&=\sum\limits_{\mathclap{\substack{R_\diag\in\sigma\\
R_\diag\text{ diagonal}}}}q(R_\diag,R_1,R_1^{-1})\cdot
|R_\diag(C)|\text{ and}\\
\text{where }|R_\diag(C)|&=\sum\limits_{\mathclap{\substack{R_2\in\sigma\\
R_2(C)= R_\diag(C) R_2(C) R_\diag(C)}}}q(R_\diag,R_2,R_2^{-1}).
\end{align*}
We show the correctness. The first equation for $|E_1(A)|$ is correct since~$E_1(A)$ is partitioned by the colour classes contained within it.
We consider the equation for $|R_1(C)|$.
Observe that
for a fixed vertex $u\in V(C)$ with $(u,u)\in R_\diag(C)$ the number
$q(R_\diag,R_1,R_1^{-1})$ counts the number of outgoing edges
$(u,v)\in R_1(C)$.
Therefore, $q(R_\diag,R_1,R_1^{-1})\cdot |R_\diag(A)|$ is the number of edges
$(u,v)\in R_1(C)$ where $u\in V(C)$ is some vertex with $(u,u)\in R_\diag(C)$.
Summing up over $R_\diag\in\sigma$ removes the dependency of 
$R_\diag\in\sigma$.
Last but not least, consider the cardinality $|R_\diag(A)|$.
For a fixed vertex $u\in V(C)$ with $(u,u)\in R_\diag(C)$, the sum counts the
number of vertices $(u,v)$ with $(v,v)\in R_\diag(C)$. This is simply the number
of vertices $v\in V(C)$ with $(v,v)\in R_\diag(C)$, which in turn is equal to
$|R_\diag(C)|$.
\end{proof}

\begin{proof}[\textbf{Proof of Lemma~\ref{lem:op3}}]
The \DWL-algorithm for each function uses the $\addunion(\pi)$ for a set $\pi\subseteq\sigma$ defined as follows.

\ref{e:dom}. We define $\pi:=\{R\in\sigma\mid R\text{ diagonal such that }|R(C)\circ E(A)|\geq 1\}$. We prove the correctness.
Assume that for $R\in\sigma$ it holds $R(C)\subseteq\diag(\dom(E(A)))$. Then, $R\in\pi$.
On the other hand, if $R\in\pi$ then $R(C)\cap\diag(\dom(E(A)))\neq\emptyset$.
Since $(A,C)$ is a coherently coloured structure, it follows that $R(C)\subseteq\diag(\dom(E(A)))$.

\ref{e:codom}. Observe that $\codom(E(A))=\dom(E(A)^{-1})$. Using \ref{e:dom} and \cref{lem:op1} Part~\ref{e:conv}, we conclude that $\codom(E(A))$
is \DWL-computable.

\ref{e:supp} Observe that $\supp(E(A))=\dom(E(A))\cup\codom(E(A))$ and that unions are \DWL-computable by \cref{lem:op1}.\ref{e:cup},
\end{proof}

\section{Proofs and Details Omitted from Section~\ref{sec:PureDWL}}

Let $(A,C)$ be a coherently $\sigma$-coloured $\tau$-structure.
A set $U\subseteq V(A)$ is called \emph{colour-aligned}
if for each diagonal colour $R\in\sigma$ the support $\supp(R(C))$ is either a subset or is disjoint from $U$.
A colour-aligned set $U\subseteq V(A)$ is called \emph{homogeneous} if
there is a diagonal colour $R\in\sigma$ such that $U=\supp(R(C))$.
Observe that a \DWL-algorithm can decide whether $\dom(X(A))$, $\codom(X(A))$ and $\supp(X(A))$ are homogeneous
for a given symbol $X\in\sigma\cup\tau$.
A colour-aligned set $U\subseteq V(A)$ is called \emph{discrete} if
for all $u\neq v\in U$ there are diagonal colours $R_u\neq R_v\in\sigma$ such that $(u,u)\in R_u(C),(v,v)\in R_v(C)$.

A coherent configuration $C$ is called \emph{homogeneous} if
the universe $V(C)$ is homogeneous
and is called \emph{discrete} if the universe $V(C)$ is discrete.

Recall the definition of \emph{subrestrictions} from the
preliminaries. The next lemma tells us that we can compute the algebraic sketch for a subrestriction.

\begin{lem}\label{lem:sub}
There is a polynomial-time algorithm that for a given algebraic sketch $D(A)$ and a given
subset $\tilde\tau\subseteq\tau$ and diagonal relation symbol $E_U\in\tau$, computes the algebraic sketch
$D(A[\tilde\tau,U])$ where $U:=\supp(E_U(A))$.
\end{lem}

\begin{proof}
Let $\tilde A:=A[\tilde\tau,U]$.
It suffices to show that we can compute $D(\tilde A,\tilde C)$
for some coherently coloured structure $(\tilde A,\tilde C)$.
Then, by \cref{lem:canData}, we can also compute $D(\tilde A)=D(\tilde A,C(\tilde A))$.
We define the algebraic sketch of $(\tilde A,\tilde C)$
as $(\tilde\tau,\tilde\sigma,\subseteq_{\tilde\sigma,\tilde\tau},\tilde q)$
where $\tilde\sigma:=\{R\in\sigma\mid \supp(R(C))\subseteq U\}$
and where $\subseteq_{\tilde\sigma,\tilde\tau}$ and $\tilde q$
are the restrictions of $\subseteq_{\sigma,\tau}$ and $q$ to $\tilde\tau\times\tilde\sigma$
and $\tilde\sigma^3$, respectively.
\end{proof}

\begin{proof}[\textbf{Proof of Theorem~\ref{theo:pure}}]
Let $\hat M$ be polynomial-time \DWL-algorithm that computes $\CE$.
We will give a pure \DWL-algorithm $M$ that computes $\CE$.
Let $D(\hat A)=(\hat\tau,\hat\sigma,\subseteq_{\hat\sigma,\hat\tau},\hat q)$ be the algebraic sketch maintained by $\hat M$.
Analogously, we let $D(A)=(\tau,\sigma,\subseteq_{\sigma,\tau},q)$ be
the algebraic sketch maintained by the pure \DWL-algorithm
$M$.
Furthermore suppose $(\hat A_1,C(\hat A_1)),\ldots,(\hat A_m,C(\hat A_m))$ is the sequence of coherently $\hat\sigma_i$-coloured
$\hat\tau_i$-structures in the cloud of the machine $\hat M$ until the machine eventually halts.

We say that our \DWL-algorithm $M$ \emph{successfully simulates the $t$-th step} for $t\leq m$ if
\begin{enumerate}
  \item\label{pure:iso} the machine $M$ maintains a subset $\omega\subseteq\tau$ and a diagonal relation $E_U\in\tau$
  with $U:=\supp(E_U(A))$
  such that $A[\omega,U]$ is isomorphic to $\hat A_t$ up to renaming relation
  symbols, that is, 
there is a bijective function $\phi:V(A[\omega,U])\to V(\hat A_t)$ and
a one-to-one correspondence between
the relation symbols $E\in\omega$ and relation symbols $\hat E\in\hat\tau_t$, i.e., $E(A)^{\phi}=\hat E(\hat A_t)$
(the machine also maintains this one-to-one correspondence between the vocabularies),
\item\label{pure:bound} $|V(A)|\leq \sum_{i\leq t}|V(\hat A_i)|$.
\end{enumerate}
When the algorithm $M$ starts, it holds $A=\hat A_1$.
Then, we create a diagonal relation symbol $E_U\in\tau$ with $E_U(A)=\diag(V(A))$.
Now, the machine $M$ already successfully simulates the first step with $\omega:=\tau\setminus\{E_U\}$.

So assume that $M$ successfully simulates the $t$-th step.
We need to explain how to simulate step $t+1$.
By \cref{lem:sub}, we can compute $D(A[\omega,U])=D(\hat A_t)$ and therefore we can track the internal run
of $\hat M$ until the machine makes an execution of $\addpair$ or $\addcom$.
Assume that the parameter of $\addpair$ or $\addcom$ is a colour $\hat R\in\hat\sigma_t$.
By deleting all relation symbols in $\tau\setminus(\omega\cup E_U)$, we can assume
that there is a colour $R\in\sigma$ with $R(C(A))^\phi=\hat R(C(\hat A_t))$.
Then, the machine $M$ can execute $\addpair(R)$ or $\addcom(R)$ to simulate the next step.
So assume that $\addpair(\hat E)$ or $\addcom(\hat E)$ is executed by $\hat M$
where $\hat E\in\hat\tau_t$.
In this case, our \DWL-algorithm $M$ finds a relation symbol $E\in\tau$ that corresponds to $\hat E\in\hat\tau_t$
using the one-to-one correspondence from \ref{pure:iso}.
Depending on whether the execution adds pairs or contracts components we do the following.

\medskip
\textit{Case 1: $\hat M$ executes $\addpair(\hat E)$.}~\\
Let $E\in\tau$ be the relation symbol corresponding to $\hat E\in\hat\tau_t$.
If $E(A)=\emptyset$, then we are done and simulated also the $(t+1)$-th step.
Otherwise, let $R_1\in\sigma$ be a colour such that $R_1(C)\subseteq
E(A)$. A pure \DWL-algorithm can
determine such a colour from the algebraic sketch $D(A)$.
The pure \DWL-algorithm executes $\addpair(R_1)$.
Then, it creates a relation symbol $E_2\in\tau$ such that $E_2(A):=E(A)\setminus R_1(C)$.
The algorithm executes $\addpair(E_2)$ recursively.
And continues until $E_m(A)=\emptyset$ for some $m\in\NN$.
Then it creates a relation symbol $D_E$ such that $D_E(A)=D_{R_1}(A)\cup\ldots\cup D_{R_{m-1}}(A)$.
We define $\omega:=\omega\cup\{E_\lef,E_\rig,D_E\}$ and create a relation symbol $E_U'\in\tau$ with $E_U'(A)=E_U(A)\cup D_E(A)$
(which we use in place of $E_U$).

We claim that $M$ simulates the $(t+1)$-th step.
We extend the bijection $\phi:V(A[\omega,U])\to V(\hat A_{t+1})$
by matching the freshly added pairs.
The relation symbols $E_\lef,E_\rig\in\tau$ correspond to the relation symbols $E_\lef,E_\rig\in\hat\tau_{t+1}$.
The number of freshly added vertices is bounded by $|V(\hat A_{t+1})|$
and therefore we maintain $|V(A)|\leq |V(\hat A_{t+1})|+\sum_{i\leq t}|V(\hat
A_i)|=\sum_{i\leq t+1}|V(\hat A_i)|$.

\medskip
\textit{Case 2: $\hat M$ executes $\addcom(\hat E)$.}~\\
Let $E\in\tau$ be the relation symbol corresponding to $\hat E\in\hat\tau_t$.
Using \cref{lem:op1} Part~\ref{e:scc}, we create a relation symbol
$E_\scc\in\tau$ with $E_\scc(A)=E(A)^\scc$ (note that the \DWL-algorithm from \cref{lem:op1} that creates this symbol is pure).

If all strongly connected components $S\in\sccs(E_\scc(A))$ are
discrete (a \DWL-algorithm can recognise this case from the algebraic
sketch), we pick diagonal relation symbols $R_1,\ldots,R_s\in\sigma$
such that $\supp(R_1(C))\cup\ldots\cup\supp(R_s(C))$
intersects each $S\in\sccs(E_\scc(A))$ in exactly one vertex.
We create a diagonal relation symbol $E_\CS\in\tau$ such that $E_\CS(A)=R_1(C)\cup\ldots\cup R_s(C)$.
Then
$\supp(E_\CS(A))\cap S$ is a singleton for each $S\in\sccs(E(A))$.
We want to use this singleton as representation of $S$.
To make this work, we create a relation symbol $E_U'\in\tau$
with $E_U'(A)=(E_U(A)\setminus\supp(E_\scc(A)))\cup E_\CS(A)$.
For each $E\in\omega$ that corresponds to some $\hat E\in\hat\tau_t$,
we define a relation symbol $E'\in\tau$ with
\[E'(A)=E_U'(A)ZE(A)ZE_U'(A)\]
 where $Z:=E_\scc(A)\cup\diag(V(A))$.
We redefine $\omega:=\{E'\in\tau\mid E\in\omega\}$.
We claim that our \DWL-algorithm $M$ simulates the $(t+1)$-th step.
The bijection $\phi:V(A[\omega,U])\to V(\hat A_{t+1})$
identifies the unique vertex in $S\cap\supp(E_\CS(A))$ with the vertex $\phi(S)\in V(\hat A_{t+1})$.
Moreover, our \DWL-algorithm $M$ does not add fresh vertices and therefore
\[|V(A)|\leq \sum_{i\leq t}|V(\hat A_i)|\leq\sum_{i\leq t+1}|V(\hat A_i)|.\]

It remains the case, in which there is a strongly connected component $S\in\sccs(E_\scc(A))$
which is not discrete.
In this case, our pure \DWL-algorithm chooses an arbitrary non-diagonal colour $R\in\sigma$ (the lexicographically least, say) with
$R(C)\subseteq E_\scc(A)$ such that $\supp(R(C))$ is homogeneous.
Such a colour exists and can be found using the algebraic sketch
(in fact, $R\in\sigma$ is the colour of $(u,v)$ for some $u,v\in S$
that have the same colour).
Then, it contracts the elements in $\sccs(R(C))$ by applying the $\addcom(R)$-instruction.
If all strongly connected components of $E_\scc(A)$ are discrete after
this contraction, the algorithm continues with the first case.
Otherwise, it iteratively contracts $E_\scc$.

To show that the algorithm terminates after polynomially many iterations,
we need to show that $\sccs(R(C))$ contains a non-singleton and
therefore the number of elements in $E_\scc(A)$ reduces in each iteration.
We claim that for a non-diagonal colour $R(C)$ that has a homogeneous support $\supp(R(C))$ it holds that
$R(C)^\scc$ is also non-diagonal.
Since $\supp(R(C))$ is homogeneous, each vertex in $R(C)$ (viewed as a
directed graph) has the same in-degree and out-degree. This implies
that every edge of $R(C)$ is contained in a cycle. Since $R$ is
non-diagonal, there is a pair $(u,v)\in R(C)$ with $u\neq v$, and
since $(u,v)$ is contained in a cycle, it is also contained in a
strongly connected component. Thus $u\neq v$ belong to the same strongly
connected component of $R(C)$ which proves the claim.

\medskip
\textit{Case 3: $\hat M$ executes $\addunion(\hat \pi)$.}\\
We define a subset $\pi\subseteq\sigma$ corresponding to
$\hat\pi\subseteq\hat\sigma_t$ as follows.
By deleting all relation symbols in $\tau\setminus(\omega\cup E_U)$, we can assume
that for each colour $\hat R\in\hat\sigma_t$ there is a colour $R\in\sigma$ with
$R(C(A))^\phi=\hat R(C(\hat A_t))$.
Let $\pi:=\{R\in\sigma\mid\hat R\in\hat\pi\}$.
We execute $\addunion(\pi)$ for $\pi\subseteq\sigma$ and
update $\omega:=\omega\cup\{E_\pi\}$.

\medskip
\textit{Case 4: $\hat M$ executes $\forget(\hat E)$.}\\
We simply update $\omega:=\omega\setminus\{E\}$ for the relation symbol
$E\in\tau$ that corresponds to $\hat E\in\hat\tau_t$.

We analyse the total running time of $M$.
Since $\hat M$ is a polynomial-time \DWL-algorithm the following values
are polynomially bounded. The length $m$ of the sequence $\hat A_1,\ldots,\hat
A_m$, the size of each universe $V(\hat A_i)$ and the running time of the underlying Turing machine.
Therefore, the size of the universe $|V(A)|\leq \sum_{i\leq m}|V(\hat A_i)|$
is polynomially bounded.
The pure \DWL-algorithm  $M$ simulates $m$ steps which can each be done in polynomial time.
\end{proof}

\begin{rem}
It can be shown that \cref{theo:pure} also holds for time bounds of the form $T^{\CO(1)}$ where
$T$ is time-constructible and grows at least linear.
For example, quasipolynomial time instead of polynomial time.
\end{rem}

\section{Proofs and Details Omitted from Section~\ref{sec:normDWL}}
\label{app:normDWL}

\begin{proof}[\textbf{Proof of Lemma~\ref{lem:ccNorm}}]
\ref{e:product}.
We say two crossing colours $R_1,R_1'\in\sigma_\across$
are equivalent, written $R_1\sim R_1'$, if $\dom(R_1(A))=\dom(R_1'(A))$ and $\codom(R_1(A))=\codom(R_1'(A))$.
We denote by $[R_1]_\sim$, the equivalence class of $R_1$.
We will define a structure $[C(A)]_\sim$
by taking the union of equivalent colours.
More precisely, the structure has the relations $\sigma_1\cup\sigma_2\cup[\sigma_\across]_\sim$
where $[\sigma_\across]_\sim:=\{[R_1]_\sim\mid R_1\in\sigma_\across\}$
and $[R_1]_\sim([C(A)]_\sim):=\bigcup_{R_1'\in[R_1]_\sim}R_1'(C(A))$.
We claim $[C(A)]_\sim$ is a coherent configuration.
Let $(u,v)$ be a crossing edge with colour $[R_1]_\sim$
and let $(u,w)$ be a crossing edge with colour $[R_2]_\sim$
and let $(w,v)$ be an edge with colour $R_3$. Note that~$R_3$ is plain.
Let $Q:=\{x\mid (u,x)$ has colour $[R_2]_\sim$
and $(x,v)$ has colour $R_3\}$.
We show the number
$q=q([R_1]_\sim,[R_2]_\sim,R_3):=|Q|$ only depends on
$[R_1]_\sim,[R_2]_\sim,R_3$ (and not on $(u,v)$).
Let $Q':=\{x\mid (x,v)$ has colour $R_3\}$.
First we argue that $Q=Q'$.
To see this assume that $x\in Q'$.
Then, the colour $R^x$ of $(x,x)$ is the same as the colour $R^w$ of $(w,w)$.
Assume that $(u,u)$ has colour $R^u\in\sigma_1\cup\sigma_2$.
We know that $(u,w)$ has colour $R_2$.
Let $R_2'$ be the colour of $(u,x)$.
We have that $\dom(R_2(A))=\dom(R_2'(A))$ since the domains
$\dom(R_2(A)),\dom(R_2'(A))$ are homogeneous sets which
intersect non-trivially in $u$.
Moreover, we also have that $\codom(R_2(A))=\supp(R^w(A))=\supp(R^x(A))=\codom(R_2'(A))$.
Therefore, $R_2\sim R_2'$, which implies $x\in Q$.
Since~$Q=Q'$ and $C(A)$ is a coherent configuration it follows that 
$Q$ only depends on $R_3$.
Similarly, the intersection numbers $q([R_1]_\sim,R_2,[R_3]_\sim)$ where~$R_2$ is plain and~$R_1$ and~$R_3$ are crossing
only depend on $[R_1]_\sim,R_2$ and~$[R_3]_\sim$.
This proves the claim that $[C(A)]_\sim$ is a coherent
configuration. 
On the other hand, the coherent configuration $[C(A)]_\sim$ still refines $A$ since $A$ has no crossing edges.
Therefore, $[C(A)]_\sim$ refines $C(A)$.
By definition, $C(A)$ also refines $[C(A)]_\sim$ and therefore $[C(A)]_\sim\equiv C(A)$
which means that each equivalence class is a singleton.

\ref{e:ind}. We explain how $D(A)=D(A,C(A))$ can be computed for given $D(A_1,C_1),D(A_2,C_2)$.
The set of colours is $\sigma=(\sigma_1\cup\sigma_2)\cupdot\sigma_\across$.
For the symbolic subset relation, we have $R\subseteq_{\sigma,\tau} E$ if and only if $R\subseteq_{\sigma_i,\tau_i} E$ for some $i\in \{1,2\}$.
And for the intersection number
for $R_1,R_2,R_3\in\sigma_i,i\in\{1,2\}$ it holds that
$q(R_1,R_2,R_3)=q_i(R_1,R_2,R_3)$.
In case that $R_j\in\sigma_\across$ for some $j\in\{1,2,3\}$, we can use Part~\ref{e:product}
to determine the intersection number $q(R_1,R_2,R_3)$.
For example, for $R_1\in\sigma_1$ and~$R_2,R_3\in\sigma_\across$ we have that  $q(R_1,R_2,R_3)=|\codom(R_2(C(A)))|=|\dom(R_3(C(A)))|$.

\ref{e:eqfine}.
Assume for contradiction that there is a coherent configuration $C_i'$ refining $A$ that is coarser than $C_i$
for some $i\in \{1,2\}$.
Without loss of generality we assume that $i=1$.
Indeed, the previous construction from \ref{e:ind} shows that given the coherently coloured structures $(A_1,C_1'),(A_2,C_2)$,
we can define a coherently coloured structure $(A,C')$ where $C'$ is coarser than $C(A)$.
This is a contradiction since $C(A)$ is a coarsest coherent configuration.
\end{proof}

\begin{proof}[\textbf{Proof of Lemma~\ref{lem:comData}}]
We show \ref{com:alg}.
For a given algebraic sketch $D(A)$,
the algebraic sketch $D(A_R)$ can be computed as follows.
By \cref{lem:canData}, it suffices to compute an algebraic sketch $D(A_R,C)$ for some coherently coloured structure $(A_R,C)$
(where $C$ refines $A_R$,
but is not necessarily as coarse as possible).

Let $C_R$ be the structure that is obtained from $C(A)$ by contracting
the strongly connected components of $R(C(A))$.
The relations $R_1(C_R)$, for $R_1\in\sigma$,
are defined in the usual way. For example,
for  $S_1,S_2\in\sccs(R(C(A)))$ we let
$(S_1,S_2)\in R_1(C_R)$ if and only if there are $v_1\in S_1,v_2\in
S_2$ such that $(v_1,v_2)\in R(C(A))$.
We define an equivalence relation ``$\sim$'' on $\sigma$
and say that $R_1\sim R_1'$ if $R_1(C_R)=R_1'(C_R)$.
The key observation is that the equivalence relation can be computed in polynomial time,
when the algebraic sketch $D(A)=D(A,C(A))$ is given (we do not need access to the actual structure $C(A)$). 
More precisely, we have $R_1\sim R_1'$ if and only if $E_VR_1(C(A))E_V=E_VR_1'(C(A))E_V$
where $E_V=\diag(V(A)\setminus \dom(R(C(A))^\scc))\cup R(C(A))^\scc$.
Equivalently, $R_1\sim R_1'$ if and only if $E_VR_1(C(A))E_V$ and
$E_VR_1'(C(A))E_V$ intersect non-triv\-ially.
By $[R_1]_\sim$, we denote the equivalence class of the colour
$R_1\in\sigma$.
Let $([C]_\sim)_R$ be the coherent configuration with
universe $V(C_R)$ and vocabulary
$\sigma_R:=\{[R_1]_\sim\mid R_1\in\sigma\}$
and relations $[R_1]_\sim(([C]_\sim)_R):=R_1(C_R)$.
We define an algebraic sketch
$D(A_R,([C]_\sim)_R)=(\tau_R,\sigma_R,\subseteq_{\sigma_R,\tau_R},q_R)$ as
follows.
Let $\tau_R:=\tau$.
We say that
\[[R_1]_{\sim}\subseteq_{\sigma_R,\tau_R} E\text{ if }R_1'\subseteq_{\sigma,\tau} E\text{ for some }R_1'\in[R_1].\]
To show the correctness, we show that $[R_1]_{\sim}\subseteq_{\sigma_R,\tau_R} E$,
if and only if $R_1(C_R)\subseteq E(A_R)$.
Assume that $R_1'\subseteq_{\sigma,\tau} E$ for some $R_1'\in[R_1]_\sim$.
Then $R_1(C_R)=R_1'(C_R)\subseteq E(A_R)$. 
On the other hand, assume that $R_1(C_R)\subseteq E(A_R)$.
Therefore, $E_VR_1(C(A))E_V$ and $E(A)$ intersect non-trivially.
Then, there is a colour $R_1'\in\sigma$ such that $R_1'(C(A))\subseteq E_VR_1(C(A))E_V\cap E(A)$.
This implies $E_VR_1'(C(A))E_V$ and $E_VR_1(C(A))E_V$ intersect non-trivially and thus $R_1'\sim R_1$.
Therefore, $R_1'\subseteq_{\sigma,\tau} E$ and thus $[R_1]_{\sim}\subseteq_{\sigma_R,\tau_R} E$.
We define 
\[q_R([R_1]_\sim,
[R_2]_\sim,[R_3]_\sim):=\frac{1}{c_{R_2,R_3}}\cdot\sum\limits_{\substack{R_2'\in[R_2]_\sim,\\R_3'\in[R_3]\sim}} q(R_1,R_2',R_3')\]
where we set $c_{R_2,R_3}=1$ whenever $\codom(R_2(C(A)))$ and $\supp(\sccs(R(C(A))))$ are disjoint
and where $c_{R_2,R_3}=|S|$ for some $S\in\sccs(R(C(A)))$ whenever \[\codom(R_2(C(A)))\subseteq\supp(\sccs(R(C(A)))).\]
This is well-defined since
$C(A)$ is a coherent configuration and therefore the sizes of the strongly connected components of $R(C(A))$
coincide.
We show the correctness.
Assume that $c_{R_2,R_3}=|S|$
and let $(u,v)\in R_1(C(A))$.
The number of vertices $w\in V(A_R)$ with $(u,w)\in [R_2]_\sim(C_R)$ and $(w,v)\in [R_3]_\sim(C_R)$
is equal to the number of strongly connected components $S\in\sccs(R(C(A)))$ for which there
are $w\in S,R_2'\in[R_2]_\sim,R_3'\in [R_3]_\sim$ such that
$(u,w)\in R_2'(C(A))$ and $(w,v)\in R_3'(C(A))$.
Instead of counting strongly connected components $S\in\sccs(R(C(A)))$,
we count vertices $v\in S\in\sccs(R(C(A)))$ and divide the result by $|S|$.
\end{proof}

\begin{proof}[\textbf{Proof of Lemma~\ref{lem:pairData}}]
We show Part \ref{pair:alg}.
For a given algebraic sketch $D(A)$,
the algebraic sketch $D(A^\omega)=(\tau^\omega,\sigma^\omega,\subseteq_{\sigma^\omega,\tau^\omega},q^\omega)$ is constructed as follows.
\[\tau^\omega:=(\tau\cup\{E_\lef,E_\rig\})\cupdot\{D_R\mid R\in\omega\}.\]
First, we describe the coherent configuration $C^\omega$ of $A^\omega$
by describing its relations $R(C^\omega)$ for $R\in\sigma^\omega$.
Later in the proof, we will express the intersection numbers of $C^\omega$
in terms of the old intersection numbers.
In the following let $v=(v_1,v_2)\in V(A^\omega)$ be a pair-vertex for a crossing edge.
We define $p_i(v):=v_i$ for both $i\in \{1,2\}$.
For a plain vertex $v\in V(A)$, we define $p_i(v)=v$ for both $i\in \{1,2\}$.
For a pair $e=(u,v)\in V(A^\omega)^2$, let $R_{e}:=\{(R_{u},R_{11},R_{12},R_{21},R_{22},R_{v})\in\sigma^6\mid (p_i(u),p_j(v))\in R_{ij}(C(A)),i,j\in\{1,2\}, (p_1(u),p_2(u))\in R_u(C(A)),(p_1(v),p_2(v))\in R_v(C(A))\}$.
The set $R_e$ can be seen as the isomorphism type of the $\sigma$-coloured graph induced on the set $V_4$ where $V_4:=\{p_1(u),p_2(u),p_1(v),p_2(v)\}$.
Define
\[\sigma^\omega:=\{R_e\mid e\in V(A^\omega)^2\}.\]
The key observation is that $V_4$ intersects both sets $\CV_i(A)$ in at most two vertices.
By \cref{lem:ccNorm}.\ref{e:product}, the colour of crossing edges $\CE_\across(A)$ only depends
on the diagonal colours of the adjacent vertices.
For this reason, the set of colours $\sigma^\omega$ only depends
on the colouring of pairs which only depends on $D(A)$ and $\omega$
(and not on the actual structures $A,C(A)$).
In fact, we can compute $\sigma^\omega$
in polynomial time when $D(A)$ and $\omega\subseteq\sigma$ are given.

The set $R_{(u,v)}$ also encodes whether $u$ is plain
since $u$ is plain if and only if $(p_1(u),p_2(u))$ has a plain colour.
Similarly, $R_{(u,v)}$ also encodes whether $v$ is plain.
Therefore, the following symbolic subset relation is well-defined
\[R_{(u,v)}\subseteq_{\sigma^\omega,\tau^\omega} E\text{ if $u,v$ are plain and }R\subseteq_{\sigma,\tau} E.\]
The set $R_{(u,v)}$ also encodes whether $u=p_i(v)$ for some given $i\in \{1,2\}$
since this is the case if and only if $u=p_1(u)$ is plain and $(p_1(u),p_i(v))$ has a diagonal colour.
This ensures that the symbolic subset relation for the relation
symbols $E_\lef,E_\rig,D_R$ is
well-defined. 
\begin{align*}
R_{(u,v)}&\subseteq_{\sigma^\omega,\tau^\omega} E_\lef\text{ if $u$ is plain and }v=(u,v_2)\\
R_{(u,v)}&\subseteq_{\sigma^\omega,\tau^\omega} E_\rig\text{ if $u$ is plain and }v=(v_1,u)
\end{align*}
Similarly, we have $R_{(u,v)}\subseteq_{\sigma^\omega,\tau^\omega} D_R$
if and only if $u=v$ is not plain and $(p_1(u),p_2(u))$ has colour $R\in\sigma$.

We need to show that $C^\omega$ indeed defines a coherent configuration.
We do this by expressing the intersection numbers recursively.
We will observe that they can be expressed using the intersection numbers from $D(A)$.
For given $R_{e_1},R_{e_2},R_{e_3}$, we assume that $R_{e_1}=R_{(u,v)},R_{e_2}=R_{(u,w)},R_{e_3}=R_{(w,v)}$
for some vertex $w$, otherwise $q^\omega=0$.

\smallskip
\textit{Base Case: $u,v,w$ are plain.}
\[q^\omega(R_{e_1},R_{e_2},R_{e_3}):=q(R_1,R_2,R_3)\]
where $R_i$ is defined as colour of $e_i\in R_i(C(A))\text{ for }i\in\{1,2,3\}$.

\smallskip
\textit{Case 1: $w$ is not plain (and thus $w=(w_1,w_2)$ a pair-vertex).}
\[q^\omega(R_{e_1},R_{e_2},R_{e_3})=\prod_{i=1}^2q^\omega(R_{e_1},R_{(u,w_i)},R_{(w_i,v)}).\]
This is well defined since $R_{(u,(w_1,w_2))}=R_{(u',(w_1',w_2'))}$ implies
that $R_{(u,w_i)}=R_{(u',w_i')}$ and $R_{(w_i,v)}=R_{(w_i',v)}$ for both $i\in \{1,2\}$.
We prove the correctness. Let $Q:=\{x\mid
R_{(u,x)}=R_{(u,w)}$ and $R_{(x,v)}=R_{(w,v)}\}$
and let $Q_i:=\{x_i\mid R_{(u,x_i)}=R_{(u,w_i)}$ and $R_{(x_i,v)}=R_{(w_i,v)}\}$ for both $i\in \{1,2\}$.
By definition, $q^\omega(R_{e_1},R_{e_2},R_{e_3})=|Q|$ and $|Q_i|=q^\omega(R_{e_1},R_{(u,w_i)},R_{(w_i,v)})$ for $i\in \{1,2\}$.
We claim the function $\phi(x):=(p_1(x),p_2(x))$ is a bijection from $Q$ to $Q_1\times Q_2$.
The function is obviously injective.
We need to show surjectivity. Let $x_1\in Q_1,x_2\in Q_2$.
By definition, $R_{(u,x_i)}=R_{(u,w_i)}$ and $R_{(x_i,v)}=R_{(w_i,v)}$ for both $i\in \{1,2\}$.
In particular, $(p_1(u),x_i)$ has the same colour as $(p_1(u),w_i)$.
Since $(w_1,w_2)$ is crossing, exactly one of the edges $(p_1(u),w_i)$ is crossing
and therefore also $(x_1,x_2)$ is crossing.
Moreover, $R_{(u,x_i)}=R_{(u,w_i)}$ implies that $(x_i,x_i)$ has the same colour as $(w_i,w_i)$.
By \cref{lem:ccNorm}.\ref{e:product}, we conclude that $(x_1,x_2)$ has same colour as $(w_1,w_2)$.
Since, we added a pair-vertex for the pair $(w_1,w_2)$ and since $(x_1,x_2)$ has the same colour,
there exists a pair-vertex $x=(x_1,x_2)$.
It follows that $x\in Q$, by the definition of $R_{(u,x)}$ and $R_{(x,v)}$.

\smallskip
\textit{Case 2: $w$ is plain, but at least one of $u,v$ is not plain.}\\
Without loss of generality we assume that $u$ is not plain, otherwise we compute the intersection number
$q^\omega$ for the conversed triple $(R_{(v,u)},R_{(v,w)},R_{(w,u)})$. So assume $u=(u_1,u_2)$.
\[q^\omega(R_{e_1},R_{e_2},R_{e_3})=q^\omega(R_{(u_k,v)},R_{(u_k,w)},R_{(w,v)})\]
where $k\in\{1,2\}$ is chosen such that $(u_k,w)$ has a plain colour.
We have to show that $q^\omega$ is well defined.
The value $k\in\{1,2\}$ does not depend on $u,w=p_1(w)$ since $R_{(u,w)}=R_{(u',w')}$
implies that $(p_i(u),p_1(w))$ and $(p_i(u'),p_1(w'))$ have the same colour for both $i\in \{1,2\}$.
Moreover, $R_{(u_i,v)}$ does not depend on the choice of $u,v$
since $R_{(u,v)}=R_{(u',v')}$ implies that $(p_i(u),v)$ and $(p_i(u'),v')$ have the same colour for both $i\in \{1,2\}$.
As above, let $Q=\{x\mid R_{(u,x)}=R_{(u,w)}$ and $R_{(x,v)}=R_{(w,v)}\}$.
We define $Q'=\{x\mid R_{(u_k,x)}=R_{(u_k,w)}$ and $R_{(x,v)}=R_{(w,v)}\}$ for $k\in\{1,2\}$ such that
$(u_k,w)$ is plain.
We claim that $Q=Q'$.
Let $x\in Q$. Then, $R_{(u,x)}=R_{(u,w)}$ and $R_{(x,v)}=R_{(w,v)}$.
Since $R_{(u,x)}=R_{(u,w)}$, it follows that
also $R_{(u_i,x)}=R_{(u_i,w)}$ for both $i\in \{1,2\}$.
In particular, $R_{(u_k,x)}=R_{(u_k,w)}$
and therefore $x\in Q'$.
On the other hand, let $x\in Q'$.
Then, $R_{(u_k,x)}=R_{(u_k,w)}$ meaning that $(u_k,x)$ and $(u_k,w)$ have the same colour.
By \cref{lem:ccNorm}.\ref{e:product}, the colour of a crossing edge only
depends on the colours of the adjacent vertices and thus $R_{(u,x)}=R_{(u,w)}$.
Therefore, $x\in Q$.
\end{proof}

\begin{proof}[\textbf{Proof of Lemma~\ref{lem:almostNorm}}]
Let $\hat M$ be a polynomial-time \DWL-algorithm that decides isomorphism on $\CC$.
We will give an almost normalised \DWL-algorithm $M$ that decides isomorphism on $\CC$ in polynomial time.
By $D(\hat A)=(\hat\tau,\hat\sigma,\subseteq_{\hat\sigma,\hat\tau},\hat q)$, we denote the algebraic sketch maintained by $\hat M$.
Analogously, we write $D(A)=(\tau,\sigma,\subseteq_{\sigma,\tau},q)$
to denote the algebraic sketch maintained by the almost normalised \DWL-algorithm
$M$.
We write $(\hat A_1,C(\hat A_1)),\ldots,(\hat A_m,C(\hat A_m))$ to denote the sequence of coherently $\hat\sigma_i$-coloured
$\hat\tau_i$-structures in the cloud of the machine $\hat M$ until the machine eventually halts.

We say that our \DWL-algorithm $M$ \emph{successfully simulates step~$t$} for $t\leq m$ if
\begin{enumerate}
\item\label{c:iso} the machine $M$ maintains a subset $\omega\subseteq\tau$ and
  a diagonal relation symbol $E_U\in\tau$ with $U:=\supp(E_U(A))$ such
  that $A[\omega,U]$ is isomorphic to $\hat A_t$ up to renaming relation
  symbols, that is, 
there is a bijective function $\phi:V(A[\omega,U])\to V(\hat A_t)$ and
a one-to-one correspondence between
relation symbol $E\in\omega$ and relation symbols $\hat E\in\hat\tau_t$, i.e., $E(A)^{\phi}=\hat E(\hat A_t)$
(the machine also maintains this one-to-one correspondence between the vocabularies),
\item\label{c:bound} $|V(A)|\leq 160\sum_{i\leq t}|V(\hat A_i)|^4$,
\item\label{c:com} there is a relation symbol $E_\pa\in\tau$ such that for each $v\in U$
the set $P_\pa(v):=\{u\mid (u,v)\in E_\pa(A)\}$ is a non-empty subset of $\CV_\across(A)\setminus U$.
We require that for all $v\neq v'$ the sets $P_\pa(v),P_\pa(v')$ are disjoint.
We let $P_\pa:=\bigcup_{v\in U}P_\pa(v)\subseteq\CV_\across(A)\setminus U$, and
\item\label{c:lr} there are relation symbols $E_\lt,E_\rt\in\tau$ such that
for each $v\in P_\pa$
there are two unique vertices $p_\lt(v),p_\rt(v)$ in~$\CV_\plain(A)\setminus U$ with
$(p_\lt(v),v)\in E_\lt(A),(p_\rt(v),v)\in E_\rt(A)$ (we can think of $E_\lt,E_\rt$ as $E_\lef,E_\rig$).
We require $(p_\lt(v),p_\rt(v))\in\CE_\across(A)$ and that
for all $v\neq v'$ that $(p_\lt(v),p_\rt(v))\neq(p_\lt(v'),p_\rt(v'))$,
\end{enumerate}

We explain how to simulate the first step.
When the algorithm $M$ starts, it holds that $A=\hat A_1$ and therefore the bound in Property \ref{c:bound} holds.
By creating a relation symbol $E_U\in\tau$ with $E_U=\diag(V(A))$,
we can assume that Property \ref{c:iso} holds.
We need to explain how to ensure Property \ref{c:com} and \ref{c:lr}.
We add unique vertices $v_1^*,v_2^*$ to the structures $A_1,A_2$
which are connected to each vertex in $v_1\in \CV_1(A),v_2\in\CV_2(A)$, respectively.
(This can be done by executing $\addpair(E_U)$ and contracting the freshly added pair-vertices
belonging to the same structure.)
In a next step, we define the crossing relation \[Z:=\{(u,v_i^*)\mid u\in V(A_{3-i}),v_i^*\in \CV_i(A)\}.\]
The crossing relation $Z$ is \DWL-computable, so we can create a relation symbol
$E_Z\in\tau$ with $E_Z(A)=Z$.
Since each edge in $E_Z(A)$ is crossing, our \DWL-algorithm is
almost normalised.
Then, we execute $\addpair(E_Z)$.
Next, we can create relation symbols $E_\pa,E_\lt,E_\rt$
with the desired properties as follows.
We copy each vertex in $U$ by executing $\addpair(E_U)$. Let~$\overline{u}= (u,u)$ denote the copy of~$u$.
We let $E_\pa(A):=\{(u,(u,v_i^*))\mid u\in V(A_{3-i})\}$.
Now, $P_\pa(u)=\{(u,v_i^*)\}\subseteq V_\across(A)$, which means that we assign each vertex $u$
a unique singleton set containing one crossing pair-vertex.
Furthermore, we define $E_\lt(A):=\{(\overline{u},(u,v_i^*))\mid u\in V(A_{3-i})\}$ and
$E_\rt(A):=\{((v_i^*),(u,v_i^*))\mid u\in V(A_{3-i})\}$.
Now, $p_\lt((u,v_i^*))=\overline{u}\in \CV_\plain(A)\setminus U$ and
$p_\rt((u,v_i^*))=v_i^*\in\{v_1^*,v_2^*\}\subseteq\CV_\plain(A)\setminus U$. This means
that we assign each pair-vertex in $(u,v_i^*)\in P_\pa$ a pair $((u,u),v_i^*)$ (using the copy $\overline{u}$ instead of
$u$ ensures that $p_\rt((u,v_i^*))$ is not contained in $U$).
Finally, $M$ simulates the first step successfully. 

So assume that $M$ simulates the $t$-th step.
We need to explain how to simulate step $t+1$.
By \cref{lem:sub}, we can compute $D(A[\omega,U])=D(\hat A_t)$ and therefore we can track the internal run
of $\hat M$ until the machine performs an execution of $\addpair(\hat B)$ or $\addcom(\hat B)$.
In this case, our \DWL-algorithm finds a relation symbol $E\in\tau$ that corresponds to $\hat B\in\hat\tau_t\cup\hat\sigma_t$, i.e.,
$\hat B=\hat E\in\hat\tau_t$ and $E(A)^\phi=\hat E(\hat A_t)$ or
$\hat B=\hat R\in\hat\sigma_t$ and $E(A)^\phi=\hat R(C(\hat A_t))$.
In the latter case, the \DWL-algorithm possibly needs to create such a symbol since the isomorphism from $A[\omega,U]$ to $\hat A_t$
does not ensure that such a symbol already exists (this can be done by Lemma~\ref{lem:op1}).
Depending on whether the execution adds pairs or contracts components we do the following.

\medskip
\textit{Case 1: $\hat M$ executes $\addcom(\hat B)$.}\\
Let $E\in\tau$ be the relation symbol corresponding to $\hat B\in\hat\tau_t\cup\hat\sigma_t$.
We contract the strongly connected components $\sccs(E(A))$
by executing $\addcom(E)$.
Since contractions are not restricted, the algorithm $M$ remains almost normalised.

We claim that our \DWL-algorithm $M$ successfully simulates step~$t+1$.
We extend the bijection $\phi:V(A[\omega,U])\to V(\hat A_{t+1})$
by identifying a vertex $S\in V(A[\omega,U])$ with the vertex $\phi(S)\in V(\hat A_{t+1})$.
Property \ref{c:com} holds since for each component $S$ we have $P_\pa(S)=\bigcup_{v\in S}P_\pa(s)$
and the sets $P_\pa(S_1),P_\pa(S_2)$ are disjoint for distinct/disjoint components $S_1,S_2$.
Property \ref{c:lr} trivially holds since the vertices and relations in $V(A)\setminus U$
are unaffected by a contraction of a component
$S\in\sccs(E(A))$ with $S\subseteq U$.

\medskip
\textit{Case 2: $\hat M$ executes $\addpair(\hat B)$.}\\
Let $E\in\tau$ be the relation symbol corresponding to $\hat B\in\hat\tau_t\cup\hat\sigma_t$.
Let $X:=\dom(E(A))$ and $Y:=\codom(E(A))$.
We assume that $X,Y$ are homogeneous sets (in $C(A)$),
otherwise we can add the pairs step by step as in the proof of \cref{theo:pure}.
By \cref{lem:op3}, the sets $X,Y$ are \DWL-computable.
We define $X_\pa:=\bigcup_{x\in X}P_\pa(x)\subseteq P_\pa\subseteq\CV_\across(A)\setminus U$
and define $X_\lt:=\{p_\lt(x)\mid x\in X_\pa\}\subseteq\CV_\plain(A)\setminus U$ and analogously $X_\rt$.
Analogous, we define $Y_\pa,Y_\lt,Y_\rt$.

We can assume that $X_\pa$ is homogeneous (in $C(A)$).
If this would not be the case, we would pick a homogeneous subset $H\subseteq X_\pa$.
Since $X$ is homogeneous, the set $H$ intersects $P_\pa(x)$ non-trivially for each $x\in X$.
Then, we redefine $E_\pa$ by creating a new relation $E_\pa'$ in place of $E_\pa$ with $E_\pa'(A)=\{(u,v)\mid u\in H,(u,v)\in E_\pa(A)\}$.
Now, it holds $P_\pa'(x)=P_\pa(x)\cap H$ for each $x\in X$.
For the same reason, we may assume $Y_\pa$ to be homogeneous.

As a consequence, also $X_\lt,X_\rt,Y_\lt,Y_\rt$ are homogeneous (in $C(A)$).
Assume for contradiction that $v,v'\in X_\lt$ such that $(v,v),(v',v')$ have different colours.
Then, let $x,x'\in X_\pa$ such that $p_\lt(x)=v$ and $p_\lt(x')=v'$.
Since $C(A)$ is a coherent configuration, also $x,x'$ have different colours.

For $x\in X$ consider the (directed) bipartite graph $L(x):=\{(p_\lt(v),p_\rt(v))\mid v\in P_\pa(x)\}$
consisting of crossing pairs between $\CV_1(A)$ and $\CV_2(A)$.
We can assume that the maximum out-degree is 1, i.e.,
for all $(a_\lt,a_\rt)\neq(b_\lt,b_\rt)\in L(x)$ it holds
$a_\lt=b_\lt\implies a_\rt\neq b_\rt$.
If this would not be the case, then we do the following.
We define an equivalence relation ``$\equiv$'' on $X_\lt$.
For $c\in X_\rt$, define $a_\lt\equiv_c b_\lt$ if $(a_\lt,c),(b_\rt,c)\in L(x)$
for some $x\in X$.
We omit the index and say $a_\lt\equiv b_\lt$ if $a_\lt\equiv_c b_\lt$ for some $c\in X_\rt$.
We claim that $\equiv_c$ does not depend on the choice of $c\in X_\rt$
and therefore $\equiv_c$ equals $\equiv$.
By \cref{lem:ccNorm}.\ref{e:eqfine}, \cref{lem:comData}.\ref{com:fine} and \cref{lem:pairData}.\ref{pair:eqfine},
imply that $C(A)[\CV_i(A)]$ can not be refined by refining $C(A)[\CV_{3-i}(A)]$.
Assume $a_\lt\equiv b_\lt$, we consider the refined structure obtained by individualising
the edge $(a_\lt,b_\lt)$, then the vertices in $W:=\{c\in X_\rt\mid a_\lt\equiv_c b_\lt\}$
have a different colour than vertices in $W\setminus X_\rt$. Since $W\neq\emptyset$
and $X_\rt$ is homogeneous, it follows that $W=X_\rt$.
This means $a_\lt\equiv_c b_\lt$ for all $c\in X_\rt$, which proves the claim.
Next, we will add vertices to the universe corresponding to the equivalence classes $[p_\lt(v)]_\equiv$.
We let $E_{X_\lt}$ be a relation symbol with $E_{X_\lt}(A)=\diag(X_\lt)$
and execute $\addpair(E_{X_\lt})$.
We want to use the equivalence classes $[p_\lt(v)]_\equiv$ instead of $p_\lt(v)$ in order to reduce
the maximum out-degree.
Let $E_{\equiv_\lt}$ be a relation symbol with $E_{\equiv_\lt}(A)=\{((u,u),(v,v))\mid u\equiv_\lt v\}$
and execute $\addcom(E_{\equiv_\lt})$.
Next, we redefine $E_\lt$ by creating an new relation symbol $E_\lt'$ where we relate $([p_\lt(v)]_\equiv,v)\in E_\lt'(A)$
instead of $(p_\lt(v),v)$.
Now, the graph $L'(x)=\{([p_\lt(v)]_\equiv,p_\rt(v))\mid v\in P_\pa(x)\}$
has the desired properties.
With the same argument, we can assume that the in-degree of $L(x)$ is 1
and therefore $L(x)$ is a (directed) matching or a disjoint union of cycles and paths.
Indeed, we can assume that $L(x)$ is a (directed) matching, otherwise we do the following.
We copy each vertex in $X_\lt$ by executing $\addpair(E_{X_\lt})$
for a relation symbol with $E_{X_\lt}=\diag(X_\lt)$.
Then, we use the copies instead of $X_\lt$. More precisely,
we redefine $E_\lt$ by creating a relation symbol $E_\lt'\in\tau$
with $E_\lt'(A)=\{((u,u),v)\mid (u,v)\in E_\lt(A)\}$.
This way, we ensure that $X_\lt,X_\rt$ are disjoint and thus $L(x)$
is a (directed) matching.

We define a relation~$Z:=$
\[\{(x_i,y_i)\in (X_\lt\cup X_\rt)\times (Y_\lt\cup Y_\rt)\mid x_i,y_i\in\CV_i(A),i\in\{1,2\}\}.\]
The relation $Z$ is \DWL-computable, so we create a relation symbol
$E_Z\in\tau$ with $E_Z(A)=Z$.

Next, we define~$\CZ:=$ \[\{((x_i,y_i),(x_j,y_j))\in Z^2\mid x_i\in\CV_i(A),x_j\in\CV_j(A),i\neq j\}.\]
The relation $\CZ\subseteq V(A)^2$ is \DWL-computable, so we create
a relation symbol $E_\CZ\in\tau$ with $E_\CZ(A)=\CZ$.
Moreover, $E_\CZ(A)$ consists of crossing edges and therefore our \DWL-algorithm $M$
remains almost normalised.
We execute $\addpair(E_\CZ)$.

For $v,v'\in X_\pa$,
we say that $v\sim_Xv'$ if there is a $x\in X$ such that $v,v'\in P_\pa(x)$.
This defines an equivalence relation since $P_\pa(x),P_\pa(x')$ are pairwise disjoint by Property \ref{c:com}.
Analogous, we define an equivalence relation for $Y_\pa$.
Observe that there is a one-to-one correspondence between the equivalence classes of $X_\pa$
and the set $X$.
We claim that
for $x_i\in X_\lt,x_j\in X_\rt$, there is a unique vertex $v^*=:\pai_X\{x_i,x_j\}$ such that $\{p_\lt(v^*),p_\rt(v^*)\}=\{x_i,x_j\}$.
First, we explain why at least one vertex exists and then we explain why it is unique.
By definition of $X_\lt,X_\rt$ there are some vertices $x_i\in X_\lt$ and $x_j\in X_\lt$
for which $v^*$ exists.
So assume for contradiction that there is a pair $(x_i',x_j')$ for which such a vertex not exists.
Then $(x_i',x_j')$ has a different colour than $(x_i,x_j)$, contradicting
\cref{lem:ccNorm}.\ref{e:product}.
To show uniqueness, we use Property \ref{c:lr}
saying that $(p_\lt(v),p_\rt(v))\neq(p_\lt(v'),p_\rt(v'))$ for $v\neq v'$.
Moreover, $X_\lt,X_\rt$ are disjoint and therefore $v^*$ is unique, which proves the claim.

We define an equivalence relation ``$\sim$'' on the set $\CZ$
and say that $((x_i,y_i),(x_j,y_j))\sim((x_i',y_i'),(x_j',y_j'))$
$\pai_X\{x_i,x_j\}\sim_X\pai_X\{x_i',x_j'\}$ and $\pai_Y\{y_i,y_j\}\sim_Y\pai_Y\{y_i',y_j'\}$.
Also this equivalence relation is \DWL-computable, so can create a relation symbol $E_\sim$
such that $E_\sim(A)=\sim$.
We want to add a fresh vertex, for each equivalence class in $[\CZ]_\sim$.
We create a copy $\CZ'$ of $\CZ$ by executing $\addpair(E_\CZ)$ again.
There is a one-to-one correspondence between elements $((x_i,y_i),(x_j,y_j))\in\CZ$
and elements $((x_i,y_i),(x_j,y_j))'\in\CZ'$.
We contract the equivalence classes in the copy $\CZ'$ by executing $\addcom(E_{\sim'})$
where $E_{\sim'}(A)=\{(z_1',z_2')\in\CZ'\times\CZ'\mid z_1\sim z_2\}$.
Clearly, the \DWL-algorithm $M$ remains almost normalised.

We define a function $\psi:E(A)\to V(A)$
by mapping an pair $e=(\pai_X\{x_i,x_j\},\pai_Y\{y_i,y_j\})\in E(A)$ to a vertex $v=\psi(e)\in D_{\CZ'}(A)$ that represents
the equivalence class $[((x_i,y_i),(x_j,y_j))']_{\sim'}$.
By the definition of the equivalence class, this function is injective.
We extend $E_U$ and create a relation symbol $E_U'\in\tau$ such that
$E_U'(A)=E_U(A)\cup\diag(\image(\psi))$.

Now, we claim that our \DWL-algorithm $M$ successfully simulates step~$t+1$.
We extend the bijection $\phi:V(A[\omega,U])\to V(\hat A_{t+1})$
by mapping an equivalence class $v\in V(A)$ to the freshly added vertex $\phi(\psi^{-1}(v))\in V(\hat A_{t+1})\setminus V(\hat A_t)$.
We have to explain, how to define $E_\pa,E_\lt,E_\rt$ for the fresh vertices $v=[z']_{\sim'}$.
We update $E_\pa$ by setting $E_\pa(A)\leftarrow E_\pa(A)\cup\{(z,[z']_{\sim'})\in\CZ\times\CZ'\}$.
To update the relation symbols $E_\lt,E_\rt$, we use the relations $E_\lef$ and $E_\rig$
and update $E_\lt(A)\leftarrow E_\lt(A)\cup\{((x_i,y_i),z)\mid z=((x_i,y_i),(x_j,y_j))\}$
and $E_\rt(A)\leftarrow E_\rt(A)\cup\{(z,(x_j,y_j))\mid z=((x_i,y_i),(x_j,y_j))\}$.

For \ref{c:bound}, we have to show that $|V(A)|\leq 160\sum_{i\leq t+1}|V(\hat A_i)|^4$.
The total number of vertices added by $M$ is linearly bounded in $|\CZ|$.
More precisely, our algorithm $M$ added at most $10|\CZ|$ vertices.
Clearly, $|V(\hat A_{t+1})|\geq|E(A)|\geq \frac{1}{2}(|X|+|Y|)$ since $E(A)$ is a biregular graph between $X$ and $Y$.
On the other hand, $|\CZ|\leq |Z|^2=(|X_\lt|+|X_\rt|)^2(|Y_\lt|+|Y_\rt|)^2$.
Since the bipartite graph $L(x)$ between $X_\lt$ and $X_\rt$ is a (directed) matching, we have $|L(x)|=|X_\lt|=|X_\rt|$.
By \cref{lem:ccNorm}.\ref{e:product}, the crossing edges between $X_\lt$ and $X_\rt$ have the same colour
and therefore $|X_\lt||X_\rt|=\sum_{x\in X}|L(x)|$.
We conclude $|X_\lt|,|X_\rt|\leq |X|$ and similar $|Y_\lt|,|Y_\rt|\leq |Y|$.
In total $|\CZ|\leq (2|X|)^2(2|Y|)^2\leq (|X|+|Y|)^4\leq 16|E(A)|^4\leq 16|\hat A_{t+1}|^4$.
This leads to $|V(A)|\leq 10|\CZ|+ 160\sum_{i\leq t}|V(\hat A_i)|^4\leq
160\sum_{i\leq t+1}|V(\hat A_i)|^4$.

\medskip
\textit{Case 3: $\hat M$ executes $\addunion(\hat \pi)$.}\\
We simply execute $\addunion(\pi)$ for some $\pi\subseteq\sigma$ that corresponds to $\hat\pi\subseteq\hat\sigma$.
Then, we update $\omega:=\omega\cup\{E_\pi\}$.

\medskip
\textit{Case 4: $\hat M$ executes $\forget(\hat E)$.}\\
We simply update $\omega:=\omega\setminus\{E\}$ for some $E\in\tau$ that corresponds to $\hat E\in\hat\tau_t$.
The fact, that we do not execute the $\forget$-command ensures that $\CV_i(A)$ remains connected for
both $i\in \{1,2\}$ (this is required for almost normalised \DWL-algorithms).
\end{proof}

\begin{proof}[\textbf{Proof of Lemma~\ref{lem:norm}}]
Let $\hat M$ be polynomial-time \DWL-algorithm that decides isomorphism on $\CC$ in polynomial time.
By \cref{lem:almostNorm}, we can assume that $\hat M$ is almost normalised.
In fact, the construction from \cref{theo:pure}
preserves that the algorithm is almost normalised,
so we can assume that $\hat M$ is almost normalised and pure
at the same time. We will give a normalised \DWL-algorithm $M$ that decides isomorphism on $\CC$ in polynomial time.
By $D(\hat A)=(\hat\tau,\hat\sigma,\subseteq_{\hat\sigma,\hat\tau},\hat q)$, we denote the algebraic sketch maintained by $\hat M$.
Analogously, we write $D(A)=(\tau,\sigma,\subseteq_{\sigma,\tau},q)$
to denote the algebraic sketch maintained by the normalised \DWL-algorithm
$M$.
We write $(\hat A_1,C(\hat A_1)),\ldots,(\hat A_m,C(\hat A_m))$ to denote the sequence of coherently $\hat\sigma_i$-coloured
$\hat\tau_i$-structures in the cloud of the machine $\hat M$ until the machine eventually halts.

We say that our \DWL-algorithm $M$ \emph{successfully simulates the $t$-th step} for $t\leq m$ if
\begin{enumerate}
  \item\label{norm:iso} the machine $M$ ensures
  that the structure $A$ is isomorphic to $\hat A_t[\CV_\plain(\hat A_t)]$ up to renaming relation
  symbols, that is, 
there is a bijective function $\phi:V(A)\to V(\CV_\plain(\hat A_t)])$ and
a one-to-one correspondence between
relation symbol $E\in\tau$ and relation symbols $\hat E\in\hat\tau_t$, i.e.,
$E(A)^{\phi}=\hat E(\hat A_t[\CV_\plain(\hat A_t)])$,
\item\label{norm:bound} $|V(A)|\leq|V(\hat A_t)|$,
\item\label{norm:sketch} the algebraic sketch $D(\hat A_t)$ is already computed,
\end{enumerate}

When the algorithm $M$ starts, it holds $A=\hat A_1$ and $D(A)=D(\hat A_t)$ therefore $M$ successfully simulates the first step.
So assume that $M$ simulates the $t$-th step.

We need to explain how to simulate step the $t+1$.
Since we already computed $D(\hat A_t)$, we can track the internal run
of $\hat M$ until the machine performs an execution of $\addpair(\hat R)$, $\addcom(\hat R)$, $\addunion(\hat\pi)$ or $\forget(\hat E)$.

\medskip
\textit{Case 1: $\hat M$ executes $\addpair(\hat R)$}\\
We have to consider two cases.
If $\hat R(C(\hat A_t))\subseteq\CE_\across(\hat A_t)$, we
do not need to adapt $\phi$.
We can compute $D(\hat A_{t+1})$ for given $D(\hat A_t)$ using \cref{lem:pairData} Part~\ref{pair:alg}.

If $\hat R(C(\hat A_t))\subseteq\CE_\plain(\hat A_t)$, then
there is a colour $R\in\tau$ corresponding to $\hat
R\in\hat\tau_t$.
This follows from the fact that
adding pair-vertices for crossing edges
and contractions of components in all intermediate steps that were only simulated but truly executed actually do not refine the coherent configuration. This is formally proved in \cref{lem:pairData} Part~\ref{pair:eqfine}
and \cref{lem:comData} Part~\ref{com:fine}.

We execute $\addpair(R)$.
Clearly, we maintain Property \ref{norm:iso} and \ref{norm:bound}.
To see that we can compute $D(\hat A_{t+1})$ we again invoke \cref{lem:pairData} Part~\ref{pair:alg}.

\medskip
\textit{Case 2: $\hat M$ executes $\addcom(\hat R)$.}\\
We can compute $D(\hat A_{t+1})$ for given $D(\hat A_t)$ using
\cref{lem:comData}.\ref{com:alg}.
To maintain the bijection $\phi$, we need to contract strongly connected components.
To do so, we restrict the strongly connected components $S\in\sccs(\hat R(C(\hat
A_t)))$ to $\CV_\plain(\hat A_t)$ and define $\CS':=\{S\cap \CV_\plain(\hat
A_t)\mid S\in\sccs(\hat R(C(\hat A_t)))\}$.
We can define a relation symbol $E\in\tau$ such that each $S\in\sccs(E(A))$
is mapped to some $\phi(S)\in\CS'$.
We execute $\addcom(E)$.
Clearly, we maintain Property \ref{norm:iso} and \ref{norm:bound}.

\medskip
\textit{Case 3: $\hat M$ executes $\addunion(\hat\pi)$.}\\
In this case, we create a relation symbol $E_\pi\in\tau$
where $E_\pi(A)^\phi=E_{\hat\pi}(\hat A_{t+1})\cap\CV_\plain(\hat A_{t+1})^2$.
To ensure Property \ref{norm:sketch}, we compute $D(\hat A_{t+1})$ for given $D(\hat A_t)$ in polynomial time.

\medskip
\textit{Case 4: $\hat M$ executes $\forget(\hat E)$.}\\
Let $E\in\tau$ be the relation symbol
corresponding to $\hat E\in\hat\tau_t$.
We execute $\forget(E)$.
To ensure Property \ref{norm:sketch}, we compute $D(\hat A_{t+1})$ for given $D(\hat A_t)$ in polynomial time.
\end{proof}

\section{Proofs and Details Omitted from Section~\ref{sec:cptdwl}}

\begin{proof}[\textbf{Proof of Lemma~\ref{lem:foh2dwl}}]
  Suppose first that the formula $\phi(x_1,x_2)$ only has 3 (free or
  bound variables). Then it is equivalent to a formula of the
  infinitary 3-variable counting logic $\LC^3_{\infty\omega}$. It
  follows from the fact that the 2-dimensional Weisfeiler Leman
  algorithm decides
  $\LC^3_{\infty\omega}$-equivalence (due to \cite{caifurimm92}) that for every $\tau$-structure
  $A$ the relation $\phi(A)$ is a union of colours of the coherent
  configuration $C(A)$. Moreover, given the algebraic sketch $D(A)$
  and $\phi$, it can be decided in polynomial time which colour
  classes these are. That is, given $D(A)$ and $\phi$ we can compute
  in polynomial time a collection of colours $R_1,\ldots,R_m\in\sigma$
  such that $\phi(A)=\bigcup_{i=1}^mR_i(C(A))$ (see
  \cite{ott97}). Once we have these colours, we can use the
  $\addunion$-operation to compute $\phi(A)$ in $\DWL$.

  Now suppose that $\phi$ contains $k>3$ variables. The trick is to
  first compute the structure $A'$ obtained from $A$ by adding all
  $k$-tuples of elements and the projections of the $k$-tuples to
  their entries. In \DWL, we can compute $A'$ by repeated
  $\addpair$-operations. We can translate $\phi(x_1,x_2)$ to a
  formula $\phi'(x_1',x_2')$ that only uses 3-variables in such a way
  that $\phi(A)=\phi'(A')$. We do this by representing
  tuples of variables in $\phi$, ranging over elements of $V(A)$, by
  single variables ranging over tuples in $A'$, and using the two
  additional variables to decode the tuples. Then we can apply the
  argument above to $A'$ and $\phi'$.
\end{proof}

\end{document}